\documentclass[11pt]{article}

\usepackage[english]{babel} 

\usepackage[latin1]{inputenc} 

\usepackage[centertags]{amsmath}
\usepackage{indentfirst,amsfonts,amssymb,amsthm}
\usepackage{cite}
\usepackage[bottom=1.5cm,top=1.5cm,left=3cm,right=2cm]{geometry}

\usepackage{tikz}
\usetikzlibrary{chains,fit,shapes}

\date{}

\newcommand{\cplp}{\ensuremath{CPL^+}}
\newcommand{\cpl}{\ensuremath{CPL}}

\newcommand{\iplp}{\ensuremath{IPL^+}}
\newcommand{\fde}{\ensuremath{FDE}}
\newcommand{\fdeto}{\ensuremath{FDE^{\to}}}
\newcommand{\nel}{\ensuremath{N4}}
\newcommand{\bfour}{\ensuremath{B_4^{\to}}}
\newcommand{\pyn}{\ensuremath{\mathcal{IDM}4}}
\newcommand{\Tm}{\ensuremath{Tm}}
\newcommand{\Tqm}{\ensuremath{T4m}}
\newcommand{\Tqsm}{\ensuremath{T45m}}
\newcommand{\TBm}{\ensuremath{TBm}}
\newcommand{\TqBm}{\ensuremath{T4Bm}}
\newcommand{\Tsm}{\ensuremath{T5m}}

\newcommand{\INT}{\ensuremath{IvFDE_T}}
\newcommand{\INTq}{\ensuremath{IvFDE_{T4}}}
\newcommand{\INTqs}{\ensuremath{IvFDE_{T45}}}
\newcommand{\INTB}{\ensuremath{IvFDE_{TB}}}
\newcommand{\INTqB}{\ensuremath{IvFDE_{T4B}}}
\newcommand{\INTs}{\ensuremath{IvFDE_{T5}}}

\newcommand{\sneg}{\ensuremath{{\sim}}}
\newcommand{\sse}{\leftrightarrow}
\newcommand{\cons}{\ensuremath{{\circ}}}
\newcommand{\lfis}{\textit{LFI}s}
\newcommand{\lets}{\textit{LET}s}

\newcommand{\A}{\ensuremath{\mathcal{A}}}

\newcommand{\axK}{\textsf{(K)}}
\newcommand{\axKu}{\textsf{(K1)}}
\newcommand{\axKd}{\textsf{(K2)}}
\newcommand{\axMu}{\textsf{(M1)}}

\newcommand{\axDNu}{\textsf{(DN1)}}
\newcommand{\axDN}{\textsf{(DN)}}
\newcommand{\axdmu}{\textsf{(DM1)}}
\newcommand{\axdmd}{\textsf{(DM2)}}
\newcommand{\axdmt}{\textsf{(DM3)}}
\newcommand{\axouimp}{\textsf{(AxP)}}

\newcommand{\axT}{\textsf{(T)}}
\newcommand{\axSq}{\textsf{(4)}}
\newcommand{\axSc}{\textsf{(5)}}
\newcommand{\axB}{\textsf{(B)}}
\newcommand{\axu}{\textsf{(Ax1)}}
\newcommand{\axd}{\textsf{(Ax2)}}
\newcommand{\axt}{\textsf{(Ax3)}}
\newcommand{\axq}{\textsf{(Ax4)}}
\newcommand{\axc}{\textsf{(Ax5)}}
\newcommand{\axs}{\textsf{(Ax6)}}
\newcommand{\axst}{\textsf{(Ax7)}}
\newcommand{\axo}{\textsf{(Ax8)}}
\newcommand{\axn}{\textsf{(Ax9)}}

\newcommand{\MP}{\textsf{MP}}
\newcommand{\nec}{\textsf{NEC}}

\newcommand{\axC}{\textsf{(N1)}}
\newcommand{\axnD}{\textsf{(N2)}}
\newcommand{\axnC}{\textsf{(N3)}}
\newcommand{\axD}{\textsf{(N4)}}
\newcommand{\axNDN}{\textsf{(N5)}}
\newcommand{\axNTN}{\textsf{(N6)}}
\newcommand{\axNN}{\textsf{(N7)}}
\newcommand{\axNT}{\textsf{(N8)}}
\newcommand{\axNo}{\textsf{(N9)}}
\newcommand{\axNde}{\textsf{(N10)}}

\newcommand{\bo}{\textsf{b}}
\newcommand{\nei}{\textsf{n}}

\newcommand{\defi}{\stackrel{\text{\tiny def}}{=} }

\theoremstyle{definition}      		
\newtheorem{definition}{Definition}    		
\newtheorem{lemma}[definition]{Lemma}	  	
\newtheorem{proposition}[definition]{Proposition}	
\newtheorem{remark}[definition]{Remark}		
\newtheorem{remarks}[definition]{Remarks}		
\newtheorem{theorem}[definition]{Theorem}	 	
\newtheorem{notation}[definition]{Notation}	 	

\begin{document}

\title{Combining swap structures: the case of Paradefinite Ivlev-like modal logics based on \fde}

\author{
{\large  Marcelo E. Coniglio}\\
{\small Centre for Logic, Epistemology and the History of Science (CLE) and}\\ 
{\small Institute of Philosophy and the Humanities (IFCH)}\\
{\small University of Campinas (UNICAMP), Campinas, Brazil}\\
{\small coniglio@unicamp.br} 
 }

\maketitle


\begin{abstract}
The aim of this paper is to combine several Ivlev-like modal systems characterized by 4-valued non-deterministic matrices (Nmatrices) with \pyn,  a 4-valued expansion of Belnap-Dunn's logic \fde\ with an implication introduced by Pynko in 1999. In order to do this, we introduce a new methodology for combining logics which are characterized by means of swap structures, based on what we call {\em superposition} of snapshots. In particular, the combination  of \pyn\ with \Tm, the 4-valued Ivlev's version of {\bf KT}, will be analyzed with more details. From the semantical perspective, the idea is to combine the 4-valued swap structures (Nmatrices) for \Tm\ (and several of its extensions) with the 4-valued twist structure (logical matrix) for \pyn.  This superposition produces a universe of 6 snapshots, with 3 of them being designated. The multioperators over the new universe are defined by combining the specifications of the given swap and twist structures. This gives rise to 6 different paradefinite Ivlev-like modal logics, each one of them characterized by a 6-valued Nmatrix, and conservatively extending the original modal logic and  \pyn. This important feature allows to consider the proposed construction as a genuine technique for combining logics. In addition, it is possible to define in the combined logics a classicality operator in the sense of {\em logics of evidence and truth} ($LETs$).
A sound and complete Hilbert-style axiomatization is also presented for the 6 combined systems, as well as a Prolog program which implements the swap structures semantics for the 6 systems, producing a decision procedure for  satisfiability, refutability and validity of formulas  in these logics.
\end{abstract}

\section{Introduction} \label{sec:Intro}

As noted by R. Ballarin in~\cite[Section~1.2]{ball:10}, it was G\"odel, in his famous article ``An Interpretation of the Intuitionistic Propositional Calculus''~\cite{god:33} the first author who introduces an axiomatization of the (then recent) Lewis system {\bf S4} separating the (classical) propositional basis from the purely modal axioms and rules, namely: axioms \axK, \axT\ and \axSq, and the necessitation rule \nec: if $\varphi$ is a theorem, then $\square\varphi$ is also a theorem. Since then, and after the resonant success of Kripke semantics for modal logics, it is commonly assumed that modal logic is simply the study of systems characterized by some specific class of Kripke frames, in which axiom \axK\ and rule \nec\ are tacitly (and naturally)  assumed as being valid. However, as some authors have pointed out, \nec\ should not be unrestrictedly assumed in {\em every} modal system, that is, in {\em any} interpretation of the $\square$ modality. Consider, for instance, the following examples adapted from~\cite{con:cerro:per:20}: if $\square \varphi$ means that ``$\varphi$ is known'' (or, more specifically, ``$\varphi$ can be effectively checked by human beings as being true''), and $\varphi$ is, for instance, a tautology containing 1000 different propositional variables, then the fact that $\varphi$ is true surely does not imply that $\square\varphi$ is true, given that such tautology is materially impossible to be checked by human beings. On the other hand, if $\square\varphi$ means that ``$\varphi$ is obligatory'' (in a deontic modal system) then it makes no sense requiring the (legal or moral) obligation $\square\varphi$ of a tautology $\varphi$, given that it is impossible to infringe such obligation. Moreover, \nec\ plays an important role in deontic paradoxes. 

Because of the doubtful status of \nec\ in epistemic and deontic contexts (among others), there were proposed the {\em non-normal} modal logics, in which \nec\ is not unrestrictedly valid. Observe that, as a direct consequence of G\"odel's \axK\ and \nec, the modal operator $\square$ is congruential, that is:  if $\varphi$ and $\psi$ are logically equivalent, so are $\square\varphi$ and $\square\psi$. It is worth noting that the  congruentiality of $\square$ is a basic principle required in the standard non-normal modal systems. That is: \nec\ is dropped in such systems, but $\square$ is still congruential.  
This feature, although extremely useful from a technical point of view, can be, once again, questionable from a conceptual point of view. By following the terminology  introduced by M.~J. Cresswell in~\cite{cresswell:75}, a logical operator $\bigodot$ is {\em hyperintensional} precisely when it is not congruential, that is: $\bigodot \varphi$ is not always logically equivalent to $\bigodot \psi$ whenever $\varphi$ is logically equivalent to $\psi$.  Hyperintensionality is nowadays defined (in modal contexts) in a broader sense: hence, $\bigodot \varphi$ and $\bigodot \psi$ may not  be logically equivalent, although $\varphi$ and $\psi$ are necessarily equivalent. While the alethic $\square$ operator is usually assumed to be congruential (that is, non hyperintensional), H. Wansing claims in~\cite{wans:17} that attitudes such as `hoping', `fearing', `supposing' and imagining' are hyperintensional, see~\cite{sep-hyperintensionality}.
In the latter reference some examples of hyperintensionality are given. For instance:  Lois Lane knows that Clark Kent is in New York, but she doesn't know that Superman is too (since she is not aware that Superman is Clark Kent). Another example (that the authors take from J. Barwise and J. Perry): you can see that Mary is eating an ice cream without seeing that Mary is eating an ice cream and John is either eating chips or not (given that you may not be seeing John at all). But $\varphi$ is logically equivalent to $\varphi \land (\psi \vee \sneg \psi)$ in classical and some other non-classical logics. The previous examples show that, for instance, `knowing' and `seeing' would be hyperintensional operators.

In a series of works started in the 1970's, Yu. Ivlev proposes a novel approach to non-normal modal logics in which the modalities $\square$ and $\lozenge$ are hyperintensional. Moreover, the various modal systems he presented are not characterized by neightborhood semantics or any other formal apparatus based on {\em possible worlds}, as happens with the standard non-normal modal logics, but they are semantically characterized by finite non-deterministic matrices  (see, for instance, \cite{ivl:73}, \cite{ivl:85}, \cite{ivl:88}, \cite{ivl:13}). This approach anticipated the use of non-deterministic matrices (or Nmatrices) for non-classical logics proposed by A. Avron and I. Lev (see~\cite{avr:lev:01}, \cite{avr:lev:05}, \cite{avr:zam:11}).  Ivlev's approach was revisited and extended in~\cite{con:cerro:per:15}, \cite{con:cerro:per:20} and~\cite{omo:sku:16}. In particular, in~\cite{con:cerro:per:15} are presented the systems \Tm, \Tqm, \Tqsm, \TBm, each one characterized by a single 4-valued Nmatrix, corresponding, respectively, to weaker versions of modal systems {\bf KT}, {\bf KT4}, {\bf KT45}, and {\bf KTB}. In~\cite{con:gol:19} this semantics was extended to a class of more general Nmatrices called {\em swap structures}, in which each truth-value is a triple $z=(z_1,z_2,z_3)$ (called {\em snapshot}) such that each coordinate $z_i$ (with values in a given Boolean algebra $\mathcal{A}$) represents a truth-value for the formulas $\varphi$, $\square\varphi$ and $\square\sneg\varphi$, respectively. Swap structures were proposed in~\cite[Chapter~6]{car:con:16}  as a general non-deterministic framework to deal with paraconsistent and other  non-classical logics. Swap structures are a non-deterministic generalization of the twist-structures introduced independently by M. Fidel~\cite{fid:78} and D. Vakarelov~\cite{vak:77}, in which the snapshots are pairs and obey to an algebraic (deterministic) structure (see below). In~\cite{con:gol:19}  it was shown that the original 4-valued Nmatrices for \Tm\ and its extensions are recovered by taking snapshots over the 2-element Boolean algebra with domain $\{0, 1\}$. The four truth-values are $T=(1,1,0)$ (meaning that $\varphi$ is necessarily true), $t=(1,0,0)$ ($\varphi$ is contingently true), $f=(0,0,0)$ ($\varphi$ is contingently false) and $F=(0,0,1)$ ($\varphi$ is necessarily false). Obviously the designated values are $T$ and $t$.

Also in the 1970's, N. Belnap and J.~M. Dunn introduce in~\cite{bel:76}, \cite{bel:77}, \cite{dunn:76} the first-degree entailment logic  (\fde), a 4-valued logic nowadays known as Belnap-Dunn logic. This logic intends to formalize the way a computer deals with information. Indeed, there are basically four possibilities when a computer receives information from different sources: an item come in as asserted (`True'), as denied (`False'), both possibilities (`Both') or none of them (`None') (see~\cite{bel:76}). These four possibilities correspond to the four values of the logic \fde. It is convenient identifying these values with ordered pairs $(a,b)$  of 0-1's, obtaining respectively the pairs $(1,0)$, $(0,1)$, $(1,1)$ and $(0,0)$. These pairs, in which the first and second coordinate correspond, respectively, to asserting and denying a given item $\varphi$ of information (symbolized respectively by $\varphi$ and $\neg\varphi$), constitute a twist-structure for \fde, in which the logical operators (conjunction, disjunction and negation) are defined in a natural way: $(a,b) \land (c,d)=(a \sqcap c,b \sqcup d)$; $(a,b) \lor (c,d)=(a \sqcup c,b \sqcap d)$; and $\neg (a,b)=(b,a)$ (here, $\sqcap$ and $\sqcup$ are the infimum and supremum in the 2-element Boolean algebra $\{0,1\}$).   These operations validate the De Morgan laws, and the negation is idempotent: $\neg\neg x=x$, defining so a (bounded) De Morgan lattice.

A. Pynko~\cite{pyn:99} studied several expansions of \fde\ from the proof-theoretic and algebraic point of view, by using twist structures expanding the ones for \fde. In particular, he introduced an interesting implication over \fde\ which gives rise to a logical system called \pyn\ (see~\cite[p.~70]{pyn:99}). Pynko's algebraic study was retaken by S. Odintsov~\cite{odin:05} in the context of Nelson's logic \nel\ and its algebraic semantics. In Section~5 of such article, Odinstov proposes the 4-valued logic  \bfour, an expansion  of \pyn\ by adding a bottom $\bot$ such that $\neg\bot$ is a theorem. This logics arises  as an special model of $\nel^\bot$, the expansion of \nel\  by a bottom $\bot$ as above. As observed by Odintsov, the implication $\to$ of \bfour\ (and \pyn) can be defined, in terms of twist structures over $\{0,1\}$, as $(a,b) \to (c,d)=(a \Rightarrow c,a \sqcap d)$ (where $\Rightarrow$ is the Boolean implication). The second component of the implication  in the twist structure  reflects the fact that, in \nel\ (and also in \pyn) $\neg(\varphi \to \psi)$ is equivalent to $\varphi \wedge \neg \psi$. It should be observed that $\neg$ is hyperintensional in \pyn\ (and also in \nel\ and  $\nel^\bot$), see Remarks~\ref{Box-not-congr}(4).

In~\cite{car:rod:19} W. Carnielli and A. Rodrigues propose a novel and fruitful interpretation of \nel\ in terms of {\em evidence}, under the name of $BLE$ ({\em Basic Logic of Evidence}). Hence, the first and second components of a pair $(a,b)$ can be seen as evidence {\em in favor} of $\varphi$ and $\neg\varphi$, respectively, being this equivalent to evidence in favor of $\varphi$ and {\em against} $\varphi$, respectively, or even  evidence against $\neg\varphi$ and $\varphi$, respectively. They consider in a series of papers (starting in~\cite{kolkata}) a family of {\em logics of evidence and truth} (\lets) based on \fde, \nel, \bfour\ and other extensions of \fde, by adding a {\em classicality} unary connective $\cons$ which allows to locally recover the validity of the explosion law and the excluded middle, in a similar way as the explosion law is locally recovered in the {\em logics of formal inconsistency} (\lfis, see for instance~\cite{car:con:16}). Namely, in \lets\ $\cons\varphi$ derives $\varphi \vee\neg\varphi$, and everything follows from $\varphi \land\neg\varphi\land\cons\varphi$.  Observe that, by definition, all these systems are both paraconsistent and paracomplete w.r.t. $\neg$. Such kind of logics are called {\em paradefinite}, or {\em paranormal}, or {\em non-alethic}.

The main objective of the present paper is to combine  the 4-valued Ivlev-like modal systems studied in~\cite{con:cerro:per:15} with extensions of \fde\ (considered as a logic of information). In order to do this, we introduce a new methodology for combining logics which are characterized by means of swap structures. In particular, the combination \INT\ of \pyn\ with \Tm, the 4-valued Ivlev's version of {\bf KT}, will be analyzed with more details. From the semantical perspective, the idea is to combine the modal 4-valued swap structures (Nmatrices) introduced  in~\cite{con:gol:19} with the 4-valued twist structure (logical matrix) for \pyn.  This combination is based on what we call {\em superposition} of snapshots. In this specific example, it produces 4-dimensional snapshots  $z=(z_1,z_2,z_3,z_4)$ such that each coordinate $z_i$ (with values 0 or 1) represents a possible truth-value for the formulas $\varphi$, $\square\varphi$, $\square\sneg\varphi$ and $\neg \varphi$, respectively, where $\sneg$ is the classical negation of \Tm\ and $\neg$ is the (paraconsistent and paracomplete) negation of \pyn. That is, the original 3-dimensional snapshots for \Tm\ are superposed with the 2-dimensional snapshots of \pyn, thus enriching the former ones with a fourth coordinate, of an epistemic character.  As it will be analized in Section~\ref{sectINT}, this superposition produces a universe of 6 snapshots, with 3 of them being designated (the ones with first coordinate 1). The multioperators over the new universe are defined by combining the specifications of the given swap and twist structures, obtaining so a 6-valued characteristic Nmatrix for each one of the 4 Ivlev-like modal logics mentioned above, plus two new ones (the Ivlev-like version of {\bf KT4B} and {\bf KT5}). This gives rise to 6 different paradefinite Ivlev-like modal logics, each one of them conservatively extending the original modal logic and  \pyn. This important feature allows us to consider the proposed construction as a genuine technique for combining logics. In addition, it is possible to define in the combined logics a classicality operator in the sense of \lets\ mentioned above. This means that the combined logics  are modal \lets.

A sound and complete Hilbert-style axiomatization is also presented for the 6 combined systems, by gathering together the axioms of the original systems plus some suitable bridge principles (axioms in the mixed language) establishing some required interactions between the logical operators. Clearly, each of these systems is decidable by its characteristic 6-valued Nmatrix. Moreover, in the Appendix we present a Prolog program which implements the swap structures semantics for the 6 systems. This program leads to a decision procedure for  satisfiability, refutability and validity of formulas  in these logics, by constructing the complete truth-tables, models and countermodels of a given formula.

\section{Preliminaries}

In this section, some basic technical notions to be used along the paper
 will be introduced.

A propositional signature is a denumerable family $\Theta=(\Theta_n)_{n \geq 0}$ of pairwise disjoint sets. The  elements of each set $\Theta_n$ are called {\em $n$-ary connectives}.
The (absolutely) free algebra over $\Theta$ generated by a denumerable set $\mathcal{V}=\{p_1,p_2,\ldots\}$ of  propositional variables is called {\em the algebra of formulas over $\Theta$}, and will be denoted by $For(\Theta)$. When $|\Theta|\defi \bigcup_{n \geq 0} \Theta_n$ is finite then, when there is no risk of confusion about the arity of the connectives,  $\Theta$ will be simply represented by $|\Theta|$. Given two finite signatures $\Theta_1$ and $\Theta_2$ then $\Theta_1 \cup \Theta_2$ and $\Theta_1 \cap\Theta_2$ are also finite signatures.

A {\em $\Theta$-algebra}, or an {\em algebra over $\Theta$}, is a pair $\mathcal{A}=\langle A,\mathcal{O}\rangle$ such that  $A$ is a non-empty set called the {\em universe} or {\em domain} of $\mathcal{A}$ and, if $\#\in\Theta_{n}$, then $\mathcal{O}(\#)$ is a function from $A^{n}$ to $A$. A {\em  logical matrix} over $\Theta$ is a triple $\mathcal{M}=\langle A, D,\mathcal{O}\rangle$  such that $\langle A,\mathcal{O}\rangle$  is a $\Theta$-algebra and $\emptyset \neq D \subseteq A$. A {\em valuation} over $\mathcal{M}$ is a homomorphism  $v:For(\Theta) \to \mathcal{A}$ of $\Theta$-algebras, that is, a function $v:For(\Theta) \to A$ such that $v(\#(\varphi_1,\ldots,\varphi_n))=\mathcal{O}(\#)(v(\varphi_1),\ldots,v(\varphi_n))$ for $\# \in \Theta_n$ and $\varphi_1,\ldots,\varphi_n \in For(\Theta)$. The logic generated by $\mathcal{M}$ is given as follows: $\Gamma \models_{\mathcal{M}}\varphi$ iff, for every valuation $v$ over $\mathcal{M}$, if  $v(\gamma)\in D$ for every $\gamma \in \Gamma$ then $v(\varphi)\in D$. Given a class $\mathbb{M}$ of matrices over $\Theta$, the logic generated by $\mathbb{M}$ is defined as follows: $\Gamma\models_{\mathbb{M}}\varphi$ iff $\Gamma \models_{\mathcal{M}}\varphi$ for every $\mathcal{M} \in \mathbb{M}$.

A {\em $\Theta$-multialgebra} or a {\em multialgebra over $\Theta$}, is   a pair $\mathcal{A}=\langle A,\mathcal{O}\rangle$  such that  $A$ is a non-empty set called the {\em universe} or {\em domain} of $\mathcal{A}$ and, if $\#\in\Theta_{n}$, then $\mathcal{O}(\#)$ is a function from $A^{n}$ to $\wp(A)\setminus\{\emptyset\}$ (the set of non-empty subsets of $A$). A {\em non-deterministic matrix} (or an {\em Nmatrix}) over $\Theta$ is a triple $\mathcal{M}=\langle A, D,\mathcal{O}\rangle$  such that $\langle A,\mathcal{O}\rangle$  is a $\Theta$-multialgebra and $\emptyset \neq D \subseteq A$. A {\em valuation} over $\mathcal{M}$ is a function $v:For(\Theta) \to A$ such that $v(\#(\varphi_1,\ldots,\varphi_n)) \in \mathcal{O}(\#)(v(\varphi_1),\ldots,v(\varphi_n))$ for $\# \in \Theta_n$ and $\varphi_1,\ldots,\varphi_n \in For(\Theta)$.
The logic generated by an Nmatrix (as well as by a class of Nmatrices) is defined as in the case of logical matrices, but now by considering valuations over Nmatrices instead of matrices.

It is worth noting that a $\Sigma$-algebra can be seen as a $\Sigma$-multialgebra such that every set $\mathcal{O}(\#)(z)$ is a singleton, for every $\# \in \Theta_n$ and every $z=(z_1,\ldots,z_n) \in A^n$. Hence, a logical matrix $\mathcal{M}$ is, in particular, an Nmatrix, and any valuation over $\mathcal{M}$ (as a matrix) is a valuation over it (seen as an Nmatrix).  These facts will be used along the paper.

\section{The 4-valued logic \pyn}

In this section the logic \pyn, the expansion of Belnap-Dunn's logic \fde\ by adding an implication mentioned in the Introduction,  will be briefly described.

\begin{definition} [The logic \pyn\ as a matrix logic] \label{IDM4-def}
Let $\Sigma_{\fde}=\{\land,\lor,\neg\}$ be the signature of \fde, and let $\Sigma_4=\{\land,\lor,\to,\neg\}$ in which $\to$ is a binary connective representing an implication.
The logic \pyn\ is given by the logical matrix $\mathcal{M}_4=\langle A_4,D_4,\mathcal{O}_4^\to\rangle$  over $\Sigma_4$, where $A_4=\{{\bf 1},\bo,\nei,{\bf 0}\}$ and $D_4=\{{\bf 1},\bo\}$, and the operations are defined as follows (by simplicity, we will identify each connective $\#$ with its interpretation $\mathcal{O}_4^\to(\#)$):
\begin{center}
\begin{tabular}{|c|c|c|c|c|}
\hline
 $\land$ & {\bf 1}   & \bo  & \nei & {\bf 0} \\
 \hline \hline
    {\bf 1}     & {\bf 1}   & \bo 		& \nei 	& {\bf 0}  \\ \hline
     \bo    & \bo   & \bo 	& {\bf 0} 	& {\bf 0}  \\ \hline
     \nei    & \nei   & {\bf 0} 	& \nei & {\bf 0}  \\ \hline
     {\bf 0}    & {\bf 0}   & {\bf 0} 		& {\bf 0} 	& {\bf 0}  \\ \hline
  \end{tabular}
\hspace{0.3cm}
\begin{tabular}{|c|c|c|c|c|}
\hline
 $\lor$ & {\bf 1}   & \bo  & \nei & {\bf 0} \\
 \hline \hline
    {\bf 1}     & {\bf 1}   & {\bf 1} 		& {\bf 1} 	& {\bf 1}  \\ \hline
     \bo    & {\bf 1}   & \bo 	& {\bf 1} 	& \bo  \\ \hline
     \nei    & {\bf 1}   & {\bf 1} 	& \nei & \nei  \\ \hline
     {\bf 0}    & {\bf 1}   & \bo 		&  \nei 	& {\bf 0}  \\ \hline
  \end{tabular}
\end{center}

\

\begin{center}
\begin{tabular}{|c|c|c|c|c|}
\hline
 $\to$ & {\bf 1}   & \bo & \nei & {\bf 0} \\
 \hline \hline
    {\bf 1}     & {\bf 1}   & \bo 		& \nei 	& {\bf 0}  \\ \hline
     \bo    & {\bf 1}   & \bo 	& \nei 	& {\bf 0}  \\ \hline
     \nei    & {\bf 1}   & {\bf 1} 	& {\bf 1} & {\bf 1}  \\ \hline
     {\bf 0}    & {\bf 1}   & {\bf 1}	& {\bf 1}	& {\bf 1} \\ \hline
  \end{tabular}
\hspace{0.3cm}
\begin{tabular}{|c||c|} \hline
 $\varphi$ & $\neg \varphi$\\
 \hline \hline
    {\bf 1}   & {\bf 0}  \\ \hline
     \bo  & \bo  \\ \hline
     \nei  & \nei \\ \hline
     {\bf 0}  & {\bf 1} \\ \hline
  \end{tabular}
\end{center}
\end{definition}

\begin{remark} \label{rem-pynk}
As observed in the Introduction, the  logic \pyn\ was originally presented by A. Pynko in 1999 (see~\cite[p.~70]{pyn:99}).  Recently, this logic was proposed in~\cite[p.~165]{hazen} under the name \fdeto. It is worth noting that the reduct of \pyn\ to  $\Sigma_{\fde}$ is precisely \fde. 
\end{remark}

\

\noindent  In the context of \fde\ it is convenient to identify its truth-values with ordered pairs over ${\bf 2}=\{0,1\}$ such that ${\bf 1}\simeq(1,0)$; $\bo\simeq(1,1)$; $\nei\simeq(0,0)$;  and ${\bf 0}\simeq(0,1)$). Then, it is possible to represent the operations in  $\mathcal{M}_4$ by means of the following twist structure over {\bf 2} (here, $z=(z_1,z_2)$ and $w=(w_1,w_2)$ are in ${\bf 2} \times {\bf 2}$):

\begin{itemize}
\item[] $z \land w = (z_1\sqcap w_1, \ z_2\sqcup w_2)$;
\item[] $z \vee w = (z_1\sqcup w_1,\ z_2 \sqcap w_2)$;
\item[] $\neg z= (z_2,z_1)$;
\item[] $z \to w = (z_1\Rightarrow w_1,z_1 \sqcap w_2)$.
\end{itemize}

\noindent
Notice that $D_4=\{z \in {\bf 2} \times {\bf 2} \ : \ z_1=1\}$.

As in the case of \fde, the intuitive reading for these pairs is that the first coordinate corresponds to the assertion of a given item $\varphi$ of information. In turn, the second coordinate corresponds to  denying $\varphi$, that is, $\neg\varphi$. In the context of Carnielli and Rodrigues's logics of evidence and truth, the coordinates of such pairs correspond to evidence in favor of $\varphi$ and  against $\varphi$ (i.e., in favor of $\neg \varphi$), respectively.

The 4-valued twist structures semantics  for \pyn\ (which coincides with its 4-valued logical matrix) can be extended to a family of twist structures defined as above, but now defined over an arbitrary classical implicative lattice.\footnote{An {\em implicative lattice} is a lattice with top element $1$ expanded with an implication $\Rightarrow$ such that $x \sqcap z \leq y$ iff $z \leq x \Rightarrow y$ (which implies that the lattice is distributive). If, in addition,  $x \sqcup (x \Rightarrow y)=1$ (or, equivalently, $(x \Rightarrow y) \Rightarrow x=x$) then it is called a {\em classical implicative lattice} (see~\cite{curry:77}). The class of classical implicative lattices is the algebraic semantics of positive classical propositional logic \cplp.}

For every $\Gamma \cup \{\varphi\} \subseteq For(\Sigma_4)$ we write $\Gamma \models_{\pyn}\varphi$ to indicate that $\varphi$ is a semantical consequence of $\Gamma$ w.r.t.  the 4-valued Nmatrix  $\mathcal{M}_4$ of \pyn.

\section{Ivlev non-normal modal logics}

In this section, Ivlev's modal logics will be recalled.

As discussed in the Introduction, with the aim of considering modal systems without the validity of the necessitation rule \nec\ (that is, non-normal modal logics), on the one hand, and without requiring the congruentionality or hyperintensionality of $\square$, on the other,\footnote{That is, as we discussed in the Introduction, the logical equivalence between $\varphi$ and $\psi$ does not implies, in general, the logical equivalence between $\square \varphi$ and $\square\psi$ (see examples in Remarks~\ref{Box-not-congr}). In the terminology of algebraic logic, this means that the logic is not {\em self-extensional}.}  Yu. Ivlev proposes in the 1970's several modal systems with these features, but still having interesting properties and philosophical interpretations. Ivlev's systems are semantically characterized by finite non-deterministic matrices (or Nmatrices). Thus, Ivlev's Nmatrix semantics was one of the pioneering works on Nmatrices, afterwards formally introduced by  A. Avron and I. Lev in the 2000's.


In~\cite{con:cerro:per:15} Ivlev-like versions of modal systems {\bf KT}, {\bf KT4}, {\bf KT45} and {\bf KTB} were proposed under the names \Tm\ (originally named $Sa^+$ by Ivlev), \Tqm, \Tqsm, \TBm, respectively. Each of these systems was semantically characterized by an specific 4-valued Nmatrix, each one with the same  domain $A_m=\{T,t,f,F\}$  and designated values $D_m=\{T,t\}$. Let us consider the case of \Tm.
According to Ivlev, these 4 truth-values have the following interpretation:
\begin{itemize}
\item[-] $T$ means `necessarily true'
\item[-]  $t$ means `true but not necessarily true'
\item[-]  $f$ means `false but not necessarily false'; and
\item[-]  $F$ means `necessarily false' or `impossible'.
\end{itemize}

Before presenting the Nmatrix for \Tm, let us recall  the analysis of Ivlev's approach to  modal logic made in~\cite[Section~1]{con:cerro:per:20}. Thus, consider the modal concepts of \emph{necessarily true} (represented by  $\mathbb{N}$), \emph{possibly true} (represented by  $\mathbb{P}$) and \emph{actually (or factually) true} (represented by  $\mathbb{A}$).
Consider the modal concept of \emph{contingent} as being coextensive with the intersection of the concepts of possibly true and possibly false. By assuming the validity of axiom \axT, being necessarily true implies being actually (or factually) true, and being necessarily false implies being actually (or factually) false. Hence, Ivlev's four truth-values are obtained as follows:

\begin{itemize}
  \item[] \ $T$: necessarily (hence, actually) true;
  \item[] \  $t$: contingently and actually true;
  \item[] \  $f$: contingently and actually false;
  \item[] \  $F$: necessarily (hence, actually) false (or impossible).
\end{itemize}

These truth-values can be formalized in a modal language (by assuming, as usual, the equivalences $\sneg\square p \equiv \lozenge \sneg p$ and $\sneg\lozenge p \equiv \square \sneg p$). By {\axT, $p \wedge \square p \equiv \square p$ and $p \wedge \lozenge p \equiv p$. In turn,  `$p$ is contingent' is $C(p) \defi \lozenge p \wedge\lozenge \sneg p \equiv C(\sneg p)$. Hence:
\begin{itemize}
\item[-] `$p$ is contingently and actually true' is 
$$p \wedge C(p)= p  \wedge \lozenge p \wedge \lozenge \sneg p \equiv p  \wedge \lozenge \sneg p \equiv p  \wedge \sneg \square p.$$

\item[-] `$p$ is contingently and actually false' is 
$$\sneg p \wedge C(\sneg p) \equiv \sneg p  \wedge \lozenge p \wedge \lozenge \sneg p \equiv \sneg p  \wedge \lozenge p.$$

\end{itemize}

\noindent
In terms of $\mathbb{A}$, $\mathbb{N}$ and $\mathbb{P}$ altogether:\\

\noindent
$T$: $\square p \equiv p \wedge  \square p \wedge \lozenge p$; \\[1mm]
$t$:  $p \wedge C(p) \equiv  p  \wedge \sneg \square p \equiv p  \wedge \sneg \square p \wedge  \lozenge p$; \\[1mm]
$f$:  $\sneg p \wedge C(\sneg p) \equiv \sneg p  \wedge \lozenge p \equiv \sneg p  \wedge \sneg\square p \wedge \lozenge p$; \\[1mm]
$F$:  $\sneg\lozenge p \equiv \sneg p \wedge \sneg\square p \wedge \sneg\lozenge  p$.

\

\noindent
The following Venn-diagram taken from~\cite[Section~1]{con:cerro:per:20} displays the relationship between the truth-values and the modal concepts of \emph{actual}, \emph{necessary} and \emph{possible} as discussed above:\\[3mm]

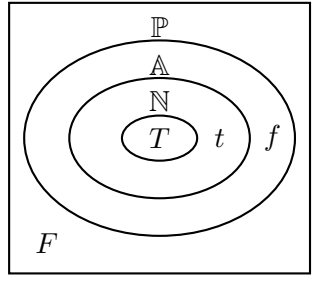
\begin{center}
\begin{tikzpicture}[thick] 
\draw (2,-1.8) rectangle (-2,1.8) node[below right] {}; 
\draw (0,0) ellipse (0.5cm and 0.3cm) node[above,shift={(0,0.7)}] {$\mathbb{A}$};
\draw (0,0) ellipse (1.2cm and 0.8cm) node[above,shift={(0,1.2)}] {$\mathbb{P}$};
\draw (0,0) ellipse (1.8cm and 1.3cm) node[shift={(0,0.5)}] {$\mathbb{N}$}; 
	\node at (0,0) {$T$}; 
	\node at (0.8,0) {$t$}; 
	\node at (1.5,0) {$f$};	
	\node at (-1.5,-1.4) {$F$}; 
\end{tikzpicture}
\end{center}

\begin{definition} [Nmatrix for Ivlev's modal logic \Tm] \label{Tm-mat-def}
Let  $\Sigma_m=\{\to, \sim, \square\}$. The modal logic \Tm\ is given by the Nmatrix $\mathcal{M}_T=\langle {A}_m,D_m,\mathcal{O}_T\rangle$  over $\Sigma_m$ such that $A_m=\{T,t,f,F\}$, $D_m=\{T,t\}$, and  where the multioperators  are defined as follows:
{\scriptsize
\begin{center}
\begin{tabular}{|c|c|c|c|c|}
\hline
 $\to$ & $T$   & $t$ & $f$ & $F$ \\
 \hline \hline
    $T$     & $T$   & $t$ 		& $f$ 	& $F$  \\ \hline
     $t$    & $T$   & $\{T, t\}$ 	& $f$ 	& $f$  \\ \hline
     $f$    & $T$   & $\{T, t\}$ 	& $\{T, t\}$ & $t$  \\ \hline
     $F$    & $T$   & $T$ 		& $T$ 	& $T$  \\ \hline
  \end{tabular}
\hspace{0.3cm}
\begin{tabular}{|c||c||c|} \hline
 $\varphi$ & $\sneg \varphi$ & $\square \varphi$ \\
 \hline \hline
    $T$   & $F$  & $\{T, t\}$   \\ \hline
     $t$  & $f$  & $\{f, F\}$    \\ \hline
     $f$  & $t$  & $\{f, F\}$  \\ \hline
     $F$  & $T$  & $\{f, F\}$   \\ \hline
  \end{tabular}
\end{center}}
\end{definition}

\

\noindent
As mentioned in the Introduction, by using the notion of swap structures introduced in~\cite[Chapter~6]{car:con:16}, Ivlev's Nmatrix was explained in an analytical way in~\cite{con:gol:19} as follows: first, by technical simplicity (given that the signatures considered by  Ivlev do not include $\lozenge$ as a primitive connective), replace the modal concept $\mathbb{P}$ of possibility by the modal concept of {\em impossibility}, denoted by $\mathbb{I}$. In formal terms, `$p$ is impossible' is represented by the sentence $\square\sneg p$, which is equivalent to $\sneg\lozenge p$. Hence, the four truth-values can be defined within this framework as follows:
\begin{itemize}
\item[-] $T:\square p \equiv p \wedge  \square p \wedge \sneg\square\sneg p$ \ \ (belongs to $\mathbb{A}$ and $\mathbb{N}$, and does not belong to $\mathbb{I}$);
\item[-] $t:  p  \wedge \sneg \square p \wedge  \sneg\square\sneg p$ \ \ (belongs to $\mathbb{A}$ and belongs to neither $\mathbb{N}$ nor $\mathbb{I}$);
\item[-] $f: \sneg p  \wedge \sneg\square p \wedge \sneg\square\sneg p$ \ \ (belongs to neither $\mathbb{A}$ nor $\mathbb{N}$ nor $\mathbb{I}$;
\item[-] $F: \square \sneg p \equiv \sneg p \wedge \sneg\square p \wedge \square\sneg  p$ \ \  (belongs to neither to $\mathbb{A}$ nor $\mathbb{N}$, and belongs to $\mathbb{I}$).
\end{itemize}
This produces a new version of the Venn diagram above:

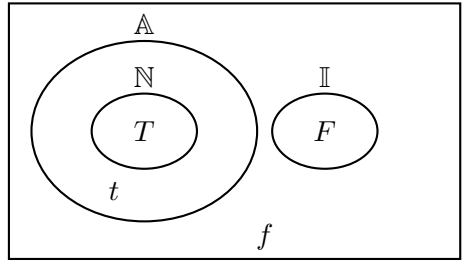
\begin{center}
\begin{tikzpicture}[thick] 
\draw (3.0,-1.7) rectangle (-3.0,1.7) node[below right] {}; 
\draw (-1.2,0) ellipse (1.5cm and 1.2cm) node[above,shift={(0,1.15)}] {$\mathbb{A}$};
\draw (-1.2,0) ellipse (0.7cm and 0.5cm) node[above,shift={(0,0.45)}] {$\mathbb{N}$};
\draw (1.2,0) ellipse (0.7cm and 0.5cm) node[above,shift={(0,0.45)}] {$\mathbb{I}$};
	\node at (-1.2,0) {$T$}; 
	\node at (1.2,0) {$F$}; 
	\node at (-1.6,-0.8) {$t$}; 
	\node at (0.4,-1.4) {$f$}; 
\end{tikzpicture}
\end{center}

Then, Ivlev's Nmatrix can be analytically described as follows: consider  truth-values as being triples $z=(z_1,z_2,z_3)$ (called {\em snapshots}) such that each coordinate  $z_i$ represents a truth-value over a Boolean algebra \A\ for the formulas $\varphi$, $\square\varphi$ and $\square\sneg\varphi$, respectively.\footnote{As we shall see, considering the 2-element Boolean algebra $\mathcal{A}_2$ with domain ${\bf 2}=\{0,1\}$ is enough. The meaning of `represents a truth-value' can be made precise, in the case of $\mathcal{A}_2$, by means of (non-truth-functional) bivaluations, as it was shown in~\cite[Section~6.4]{car:con:16} and, with a more general perspective, in~\cite[Section~5.2]{con:tol:22} and~\cite{con:rod:22}, in the context of swap structures semantics for \lfis\ and \lets. In Section~\ref{finalsect} this point will be briefly discussed.}  Because of  axiom \axT, the domain of snapshots (which are 3-dimensional truth-values) is 
$$\mathbb{B}_{\mathcal{A}} = \{(a_1, a_2, a_3) \in A^3 \ : \ a_2 \leq a_1 ~~\mbox{and}~~ a_1 \sqcap a_3 = 0\}.$$

\begin{definition} \label{NMatIvlev} Let \A\ be a Boolean algebra. A swap structure for \Tm\ is a multialgebra $\mathcal{B}_{\A} = \langle \mathbb{B}_{\mathcal{A}},\hat{\to}, \hat{\sim}, \hat{\square} \rangle$ over $\Sigma_m=\{\to, \sim, \square\}$ where the multioperators are defined as follows:\\[1mm]
(i) $z\hat{\to} w = \{u\in \mathbb{B}_{\mathcal{A}} \ : \ u_1=z_1\Rightarrow w_1, \ u_3=z_2\sqcap w_3,$\\
	\hspace*{3.5cm} $z_3 \sqcup w_2 \leq u_2 \leq (z_2 \Rightarrow w_2) \sqcap (w_3 \Rightarrow z_3) \}$;\\[1mm]
(ii) $\hat{\sneg} z = \{(\sneg z_1,z_3,z_2)\}$;\\[1mm]
(iii) $\hat{\square} z = \{u\in \mathbb{B}_{\mathcal{A}} \ : \ u_1=z_2\}$.\\[1mm]
Here, $\Rightarrow$, $\sqcup$, $\sqcap$ and $\sneg$ denote respectively the implication, supremum, infimum and Boolean complement in \A. This gives rise to a Nmatrix $\mathcal{M}_{\A}$ in which the designated values are $\mathbb{D}_{\mathcal{A}} = \{a \in \mathbb{B}_{\A} \ : \ a_1=1\}$.
\end{definition}

\noindent
In a more concise way, the multioperations can be specified as follows:

\begin{itemize}
\item[] $z \hat{\to} w = (z_1\Rightarrow w_1,\ z_3 \sqcup w_2 \leq \_ \leq (z_2 \Rightarrow w_2) \sqcap (w_3 \Rightarrow z_3),z_2\sqcap w_3)$
\item[] $\hat{\sneg} z= (\sneg z_1,z_3,z_2)$
\item[] $\hat{\square} z = (z_2,\_,\_ )$
\end{itemize}

\noindent
The symbol `$\_$' denotes that the respective position of the snapshot can be freely chosen, provided that it satisfies the restrictions (in the case of the second coordinate of $\hat{\to}$) and that  the resulting triple belongs to $\mathbb{B}_{\mathcal{A}}$. The description above for the multioperator $\hat{\#}$ in $\mathcal{B}_{\A}$  is called the {\em formal specification} of $\hat{\#}$, given that it is expresssed by means of Boolean terms of the form  $t(z,w)$ (and, in the case of the second coordinate of $\hat{\to}$, by using also the order $\leq$). Observe that the formal expressions for the twist structures semantics for \pyn\ presented after Remark~\ref{rem-pynk} is the formal specification of that operators, which are defined over pairs.
%
 
%

It is worth noting that, when $\A=\A_2$, the 2-element Boolean algebra with domain ${\bf2}$, then $\mathcal{M}_{\mathcal{A}_2}$ is nothing else than Ivlev's characteristic 4-valued Nmatrix $\mathcal{M}_T$ for \Tm\ as introduced in Definition~\ref{Tm-mat-def}, in which $T=(1,1,0)$, $t=(1,0,0)$, $f=(0,0,0)$ and $F=(0,0,1)$. This identification will be used from now on.

For every $\Gamma \cup \{\varphi\} \subseteq For(\Sigma_m)$ we write $\Gamma \models_{\Tm}\varphi$ to denote that $\Gamma$ entails  $\varphi$ w.r.t. the semantics given by the 4-valued Nmatrix $\mathcal{M}_T$. That is: for every valuation $v$ over $\mathcal{M}_T$, if $v(\gamma) \in D_m$ for every $\gamma \in \Gamma$ then $v(\varphi) \in D_m$. As proved in~\cite{con:gol:19}, this is equivalent to say that $\Gamma$ entails  $\varphi$ w.r.t. the semantics given by the class of  Nmatrices  for \Tm\ of the form $\mathcal{M}_\A$.

\begin{remark} \label{conj-disj-Tm}
In order to deal with conjunction and disjunction in \Tm, the following terms are considered: $\varphi \land \psi \defi \sneg(\varphi \to \sneg\psi)$ and $\varphi \vee \psi \defi \sneg\varphi \to \psi$. It is easy to see that the associated multioperators in the swap structure $\mathcal{B}_{\A}$ are specified as follows:

\begin{itemize}
\item[] $z \hat{\land} w = (z_1\sqcap w_1, \ z_2\sqcap w_2,\ z_3 \sqcup w_3 \leq \_ \leq (z_2 \Rightarrow w_3) \sqcap (w_2 \Rightarrow z_3))$
\item[] $z \hat{\vee} w = (z_1\sqcup w_1,\ z_2 \sqcup w_2 \leq \_ \leq (z_3 \Rightarrow w_2) \sqcap (w_3 \Rightarrow z_2),z_3\sqcap w_3)$
\end{itemize}
In particular, in $\mathcal{M}_T$ the multioperators are defined as follows:
{\scriptsize
\begin{center}
\begin{tabular}{|c|c|c|c|c|}
\hline
 $\land$ & $T$   & $t$ & $f$ & $F$ \\
 \hline \hline
    $T$     & $T$   & $t$ 		& $f$ 	& $F$  \\ \hline
     $t$    & $t$   & $t$ 	& $\{F, f\}$ 	& $F$  \\ \hline
     $f$    & $f$   & $\{F, f\}$ 	& $\{F, f\}$ & $F$  \\ \hline
     $F$    & $F$   & $F$ 		& $F$ 	& $F$  \\ \hline
  \end{tabular}
\hspace{0.3cm}
\begin{tabular}{|c|c|c|c|c|}
\hline
 $\lor$ & $T$   & $t$ & $f$ & $F$ \\
 \hline \hline
    $T$     & $T$   & $T$ 		& $T$ 	& $T$  \\ \hline
     $t$    & $T$   & $\{T, t\}$ 	& $\{T, t\}$	& $t$  \\ \hline
     $f$    & $T$   & $\{T, t\}$ 	& $f$ & $f$  \\ \hline
     $F$    & $T$   & $t$ 		& $f$ 	& $F$  \\ \hline
  \end{tabular}
\end{center}}

\

These definitions will be useful in order to construct the combination of \Tm\ with \pyn.
\end{remark}

\section{Combining swap structures: the case of modal logic \INT} \label{sectINT}

Having motivated and presented the extension \pyn\ of \fde\, and the Ivlev modal logic \Tm, in this section it will be introduced our main proposal: the  semantical  combination of Ivlev-like modal logics with extensions of \fde\ by means of a novel technique for combining logics presented by means of swap structures. This technique will be exemplified by defining the modal logic \INT, which is the result of combining Ivlev's \Tm\ with \pyn, an extension of Nelson's \nel\ (which, in turn, extends \fde). It is worth noting that, while \Tm\ is semantically characterized by a 4-valued swap structure,  \pyn\ is characterized by a 4-valued twist structure, a particular (deterministic) case of swap structures. As it will shown in Section~\ref{conserva}, \INT\ can be considered as a genuine combination of \Tm\ and \pyn, provided that it is a conservative expansion of both logics.

The basic idea of the technique is called {\em superposition} of snapshots, defined in a general setting as follows.\footnote{To simplify the presentation, and to obtain a decision procedure for the combined logic, it will be considered the combination of the finite-valued characteristic Nmatrix for each of the logics being combined, corresponding to the swap structure over the 2-element Boolean algebra. As discussed in Section~\ref{finalsect} below, the method can be applied to the whole family of Nmatrices for each of the given logics, producing a class of swap structures over Boolean algebras.} For  $i=1,2$ assume that the  logic ${\bf L}_i$, defined over a finite signature $\Theta_i$, is characterized by a swap structures semantics given by a  finite  Nmatrix $\mathcal{M}_i=\langle A_i,D_i,\mathcal{O}_i\rangle$ such that the underlying multialgebra $\mathcal{A}_i=\langle A_i,\mathcal{O}_i\rangle$ is a swap structure over the 2-element Boolean algebra $\A_2$  where the snapshots in $A_i \subseteq {\bf 2}^{n_i+1}$ represent the possible 0-1 values assigned to the components of the sequence of (schema) formulas $(\xi,\psi^i_1(\xi), \ldots,\psi^i_{n_i}(\xi))$. Here, $\xi$ is a schema variable representing an arbitrary formula over the signature $\Theta_i$, and each $\psi^i_j$ is a schema formula  over the signature $\Theta_i$ depending exclusively on $\xi$.\footnote{For instance, in ${\bf L}_1=\Tm$ we have $n_1=2$, $\psi^1_1(\xi)=\square\xi$ and $\psi^1_2(\xi)=\square\sneg\xi$. In turn,  in ${\bf L}_2=\pyn$ we have $n_2=1$ and $\psi^2_1(\xi)=\neg\xi$.}
 The superposition of snapshots of ${\bf L}_1$ and ${\bf L}_2$ is defined as the finite set  of  snapshots  representing altogether the possible 0-1 values which can be assigned to the components of the sequence of formulas $(\xi,\psi^1_1(\xi), \ldots,\psi^1_{n_1}(\xi), \psi^2_{1}(\xi), \ldots,\psi^2_{n_2}(\xi))$, where now $\xi$ represents an arbitrary formula over the signature $\Theta_1 \cup \Theta_2$.\footnote{As mentioned previously, the notion of snapshots as representing all the possible 0-1 values assigned to the formulas corresponding to its components can be rigorously formalized by means of bivaluations. In Section~\ref{finalsect} this question will be briefly discussed.} This gives rise to a finite multialgebra $\A$ over the signature $\Theta_1 \cup \Theta_2$ presented as a swap structure whose domain is the set $A \subseteq {\bf 2}^{n_1+n_2+1}$ of (superposed) snapshots. This means that the requirements for the components of the snapshots in  ${\bf L}_i$ used to construct the domain  $A_i$ ($i=1,2$) are used  together in the construction of the domain $A$.  The multioperator $\#^\A$ over $A$ associated to a connective $\#$ of $\Theta_i$ preserves the formal specifications  from its interpretation $\#^{i}$ in $\A_i$ (when $i=2$, subscripts corresponding to the coordinates $\neq 1$ of the snapshots need to be adjusted in the specifications). In case  $\#$ is shared, that is, it belongs to $\Theta_1 \cap \Theta_2$, it is assumed that  the formal specification of the first coordinate of $\#^1$ and $\#^2$ coincide (a basic coherence requirement for defining the combination of the swap structures).\footnote{For instance, the formal specification for the first coordinate of the interpretation of $\to$ in the (N)matrices of \pyn\ and \Tm\ is the same, namely: $z_1 \Rightarrow w_1$. Then, $\to$ can be shared in the combination of \pyn\ and \Tm.}
 
In the examples to be analyzed in this paper, 3-dimensional snapshots for logic ${\bf L}_1 \defi \Tm$  (as well as its axiomatic extensions), representing 0-1 values for formulas $(\xi,\square \xi, \square\sneg \xi)$ over $\Sigma_m$, will  be superposed to 2-dimensional snapshots for logic ${\bf L}_2 \defi \pyn$, representing 0-1 values for formulas $(\xi,\neg \xi)$ over $\Sigma_4$.
The superposition of snapshots in this case can be represented as follows (here, snapshots for $\Tm$ are represented horizontally, while snapshots for \pyn\ are represented vertically):

\begin{center}

\begin{tikzpicture}
\tikzstyle{every path}=[very thick]

\edef\sizetape{0.7cm}
\tikzstyle{tmtape}=[draw,minimum size=\sizetape]
\tikzstyle{tmhead}=[draw,minimum size=\sizetape]
\tikzstyle{tmhead1}=[on chain=1]

\begin{scope}[start chain=1 going right,node distance=-0.15mm]
    \node [on chain=1,tmtape] (input) {$\ \ \xi \ \ $};
    \node [on chain=1,tmtape] {$\ \square\xi \ $};
    \node [on chain=1,tmtape] {$\square\sneg\xi$};
    \node [on chain=1] {\textbf{Snapshot for \Tm}};
\end{scope}

\node [tmhead,yshift=-.3cm] at (input.south) (head) {$\ \neg\xi \ $};

\node [tmhead1,yshift=-.9cm] at (input.south) (head){\hspace*{-16mm}\textbf{Snapshot}};

\node [tmhead1,yshift=-1.4cm] at (input.south) (head){\hspace*{-16mm}\textbf{for \pyn}};

\end{tikzpicture}
\end{center}

\noindent
Hence, the snapshots obtained by superposition are 4-tuples $z=(z_1,z_2,z_3,z_4)$ such that each coordinate $z_i$ (with values 0 or 1) represents a possible truth-value for formulas $\xi$, $\square\xi$, $\square\sneg\xi$ and $\neg \xi$, respectively, over the signature $\Sigma\defi \Sigma_m \cup \Sigma_4=\{\land, \lor, \to,$\\ 
$\neg,\sneg, \square\}$. Observe that $\to$ is the only shared connective (and, as observed above, this connective is shareable in these logics). As it was previously mentioned, besides the requirements for the components of the snapshots (inherited from the  original swap structures)\footnote{For instance: $z_2 \leq z_1$, coming from \Tm\ (observe that \pyn\ has no restrictions for the components of its snapshots, since $A_4={\bf 2}^2$).} required to define the domain of the combined swap structure, some additional requirements will be imposed for the superposed snapshots, motivated by conceptual considerations. 

First, \fde\ and its extensions \nel\ and \pyn\ can be interpreted in informational terms. On the other hand, the  original interpretation for the modality $\square$ in \Tm\ was given in alethic terms (hence $\square \varphi$ means that $\varphi$ is necessarily true). However, other interpretation could be given for \Tm, more compatible with an informational perspective. We can assume that $\varphi$ means that $\varphi$ was observed (or informed, or reported), while $\square \varphi$ means that $\varphi$ was confirmed (or verified, or certified). Hence, all the axioms of \Tm, including \axT, make sense with this interpretation (note that \nec, as well as the converse of \axT, would not be sound with this interpretation). Now, in \pyn\ we could give the following interpretation in informational terms: $\varphi$ means predisposition (or propensity, or inclination)  to accept $\varphi$, while $\neg\varphi$ represents predisposition (or propensity, or inclination)  to reject $\varphi$. By assuming that propensity to reject propensity to reject $\varphi$ is equivalent to propensity to accept $\varphi$, it gives an interpretation of \pyn\ compatible with the one just given for \Tm\ in informational terms. Note that all the De Morgan laws for $\neg$ in \pyn\ (including the one for implication) are compatible with this interpretation. 

Now, in principle exactly 8 snapshots should be generated by superposition, given that each original snapshot for \Tm\ splits into two ones, according to the value of $\neg\varphi$, which contains additional informational status of $\varphi$ (predisposition to accept/reject $\varphi$).
However, in case $\varphi$ is a confirmed information, which corresponds to  $T=(1,1,0)$ in \Tm, there could be no predisposition to reject $\varphi$, and so the fourth coordinate (representing $\neg\varphi$) can only be $0$. Analogously, if  confirmedly $\varphi$ is not the case (or, equivalently, $\sneg\varphi$ is confirmed), corresponding to $F=(0,0,1)$ in \Tm, there could be no propensity to accept $\varphi$ (i.e., only the propensity to reject $\varphi$ would make sense), and so the fourth coordinate (representing $\neg\varphi$) can only be $1$. This produces a universe of 6 snapshots, with exactly 3 of them being designated (the ones with first coordinate 1). As observed above, the multialgebra (swap structure) will be defined over the union $\Sigma$ of both signatures, where $\to$ is shared.

The six truth-values associated to the superposed snapshots are the following:
$$\begin{array}{ccc}
T_0=(1,1,0,0) & \ \ \ & f_0=(0,0,0,0)\\
t_0=(1,0,0,0) & \ \ \ & f_{1}=(0,0,0,1)\\
t_1=(1,0,0,1) & \ \ \ &  F_1=(0,0,1,1)
\end{array}$$
The meaning of the truth-values is the following:\\[1mm]
$T_0$: confirmed (hence observed) and without propensity to reject;\\
$t_0$: observed but not confirmed, and without propensity to reject;\\
$t_1$: observed but not confirmed, and with propensity to reject;\\
$f_0$: observed (but not confirmed) that it is not the case, and without\\
\hspace*{5mm} propensity to reject;\\
$f_1$: observed (but not confirmed) that it is not the case, and with propensity to reject;\\
$F_1$: confirmed (hence observed) that it is not the case, and  with propensity to reject.

The new concept `propensity to reject $p$' is formalized as $\neg p$. Hence, by considering, besides {\em observed} ($\mathbb{O}$), {\em confirmed} ($\mathbb{C}$) and {\em confirmed that it is not the case} ($\mathbb{CN}$), coming from the new interpretation of \Tm, the new concept of  `propensity to reject' ($\mathbb{PR}$) coming from \pyn, the truth values can be formalized  in terms of these concepts as follows:\\[1mm]

\noindent
 $T_0$: $\square p \equiv p \wedge  \square p \wedge \sneg\square\sneg p \land \sneg \neg p$ \  (belongs to $\mathbb{O}$ and $\mathbb{C}$, and belongs to neither $\mathbb{CN}$ nor  $\mathbb{PR}$);\\[1mm]
$t_0$: $p  \wedge \sneg \square p \wedge  \sneg\square\sneg p \land \sneg\neg p$ \ (belongs to $\mathbb{O}$, and belongs to neither $\mathbb{C}$ nor $\mathbb{CN}$ nor $\mathbb{PR}$);\\[1mm]
$t_1$: $p  \wedge \sneg \square p \wedge  \sneg\square\sneg p \land \neg p$ \ (belongs to $\mathbb{O}$ and $\mathbb{PR}$, and belongs to neither $\mathbb{C}$ nor $\mathbb{CN}$);\\[1mm]
$f_0$: $\sneg p  \wedge \sneg\square p \wedge \sneg\square\sneg p \land \sneg\neg p$ \  (belongs to neither $\mathbb{O}$ nor $\mathbb{C}$ nor $\mathbb{CN}$ nor $\mathbb{PR}$);\\[1mm]
$f_1$: $\sneg p  \wedge \sneg\square p \wedge \sneg\square\sneg p \land \neg p$ \  (belongs to $\mathbb{PR}$, and belongs to neither $\mathbb{O}$ nor $\mathbb{C}$ nor $\mathbb{CN}$);\\[1mm]
$F_1$: $\square \sneg p \equiv \sneg p \wedge \sneg\square p \wedge \square\sneg  p \land \neg p$ \ (belongs to $\mathbb{CN}$ and $\mathbb{PR}$, and belongs to neither $\mathbb{O}$ nor $\mathbb{C}$).\\[1mm]

This can be displayed through  the following Venn diagram:

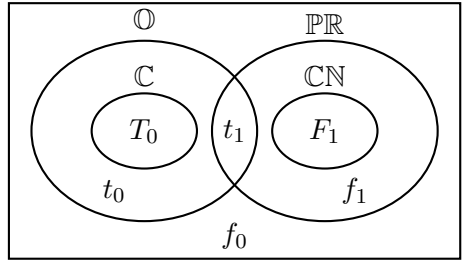
\begin{center}
\begin{tikzpicture}[thick] 
\draw (3.0,-1.7) rectangle (-3.0,1.7) node[below right] {}; 
\draw (-1.2,0) ellipse (1.5cm and 1.2cm) node[above,shift={(0,1.15)}] {$\mathbb{O}$};
\draw (1.2,0) ellipse  (1.5cm and 1.2cm) node[above,shift={(0,1.15)}] {$\mathbb{PR}$}; 
\draw (-1.2,0) ellipse (0.7cm and 0.5cm) node[above,shift={(0,0.45)}] {$\mathbb{C}$};
\draw (1.2,0) ellipse (0.7cm and 0.5cm) node[above,shift={(0,0.45)}] {$\mathbb{CN}$};
	\node at (-1.2,0) {$T_0$}; 
	\node at (0.0,0) {$t_1$}; 
	\node at (1.2,0) {$F_1$}; 
	\node at (-1.6,-0.8) {$t_0$}; 
	\node at (1.6,-0.8) {$f_1$}; 
	\node at (0.0,-1.4) {$f_0$}; 
\end{tikzpicture}
\end{center}

\begin{definition} [The combined Nmatrix for $\mathcal{M}_{\INT}$] \label{defNmatINT} \ \\ 
Let $\Sigma$ be the signature $\{\land, \lor, \to, \neg,\sneg, \square\}$.
The Nmatrix  $\mathcal{M}_{\INT}= \langle \textsc{B}_{T}, \textrm{D}, \mathcal{O}\rangle$  for \INT\ over $\Sigma$ is defined by the following swap structure:
$$\textsc{B}_{T}=\{z \in {\bf 2}^4 \ : \ z_2 \leq z_1, \ z_1 \sqcap z_3=0, \ z_2 \sqcap z_4=0   \ \mbox{ and } \ z_3 \leq z_4 \}.$$ 
That is, $\textsc{B}_{T} = \big\{T_0, \, t_0, \, t_1, \, f_0, \, f_1, \, F_1\big\}$, where these triples are defined as above.\footnote{Observe that the restrictions for the components imposed in \Tm\ are preserved in \INT\ (\pyn\ has no restrictions on the components of its snapshots). In addition, the new restrictions $z_2 \sqcap z_4=0$  and $z_3 \leq z_4$ were considered here only by conceptual reasons, not being mandatory for this process of combination of logics (see discussion at the end of Section~\ref{conserva}).}  
The set of designated elements of $\mathcal{M}_{\INT}$ is $\textrm{D}=\big\{z \in \textsc{B}_{T} \ : \ z_1=1\big\} = \big\{T_0, \, t_1, \, t_0\big\}$, hence $\textrm{ND}=\big\{z \in \textsc{B}_{T} \ : \ z_1\neq 1\big\} = \big\{f_1, \, f_0, \, F_1\big\}$ is the set of non-designated truth-values.  The multioperator $\mathcal{O}(\#)=\tilde{\#}$ interpreting the connective $\#$ is defined as follows, for  $\# \in \Sigma$ and for every $z$ and $w$ in $\textsc{B}_{T}$:

\noindent
(i) $z \tilde{\land} w = \{u\in \textsc{B}_{T} \ : \ u_1=z_1\sqcap w_1, \  u_2=z_2\sqcap w_2, \ u_4=z_4 \sqcup w_4,$\\[2mm]
	\hspace*{1.7cm}$\mbox{ and } \ z_3 \sqcup w_3 \leq u_3 \leq (z_2 \Rightarrow w_3) \sqcap (w_2 \Rightarrow z_3)\}$;\\[2mm]
(ii) $z \tilde{\lor} w = \{u\in \textsc{B}_{T} \ : \ u_1=z_1\sqcup w_1, \ u_3= z_3 \sqcap w_3,   \  u_4=z_4 \sqcap w_4,$\\[2mm]
	\hspace*{1.7cm}$\mbox{ and } \ z_2 \sqcup w_2 \leq u_2 \leq (z_3 \Rightarrow w_2) \sqcap (w_3 \Rightarrow z_2)\}$;\\[2mm]
(iii) $z\tilde{\to} w = \{u\in \textsc{B}_{T} \ : \ u_1=z_1\Rightarrow w_1, \ u_3=z_2\sqcap w_3,  \ u_4=z_1\sqcap w_4,$\\[2mm]
	\hspace*{1.7cm}$\mbox{ and } \ z_3 \sqcup w_2 \leq u_2 \leq (z_2 \Rightarrow w_2) \sqcap (w_3 \Rightarrow z_3)\}$;\\[2mm]
(iv) $\tilde{\sneg} z =   \{u\in \textsc{B}_{T} \ : \ u_1=\sneg z_1, \  u_2=z_3, \ u_3=z_2, \ \mbox{ and } \ z_2 \leq u_4 \leq \sneg z_3\}$;\\[2mm]
(v) $\tilde{\neg} z = \{(z_4,z_3,z_2,z_1)\}$;\\[2mm]
(vi) $\tilde{\square} z = \{u\in \textsc{B}_{T} \ : \ u_1=z_2\}$.
\end{definition}

\

\noindent The formal specification of the multioperations for the Nmatrix $\mathcal{M}_{\INT}$ is as follows:

\noindent
(1) $z\tilde{\land} w = (z_1\sqcap w_1, \ z_2\sqcap w_2, \ z_3 \sqcup w_3 \leq \_ \leq (z_2 \Rightarrow w_3) \sqcap (w_2 \Rightarrow z_3),\ z_4 \sqcup w_4)$;\\[1mm]
(2) $z\tilde{\lor} w = (z_1\sqcup w_1, \ z_2 \sqcup w_2 \leq \_ \leq (z_3 \Rightarrow w_2) \sqcap (w_3 \Rightarrow z_2), \ z_3 \sqcap w_3, \ z_4 \sqcap w_4)$;\\[1mm]
(3) $z \tilde{\to} w = (z_1\Rightarrow w_1,\ z_3 \sqcup w_2 \leq \_ \leq (z_2 \Rightarrow w_2) \sqcap (w_3 \Rightarrow z_3),z_2\sqcap w_3, \ z_1\sqcap w_4)$;\\[1mm]
(4) $\tilde{\sneg} z = (\sneg z_1,z_3,z_2, z_2 \leq \_ \leq \sneg z_3)$;\\[1mm]
(5) $\tilde{\neg} z = (z_4,z_3,z_2,z_1)$;\\[1mm]
(6) $\tilde{\square} z = (z_2,\_,\_ ,\_ )$.

\noindent By using the same notation adopted right after Definition~\ref{NMatIvlev},  the symbol `$\_$' means that the respective coordinate of the snapshot can be filled out by any value in {\bf 2}, provided that the resulting 4-tuple belongs to $\textsc{B}_{T}$ (that is, it is a valid snapshot).

\begin{remarks} \label{obs-swap}  \ \\
(1) By considering the projection on the first 3 coordinates, the formal specifications for the multioperators over $\Sigma_m$ coincides with the ones for \Tm. In turn, by projecting on the first and fourth coordinates, the specifications in the signature $\Sigma_4$ coincides with the ones for \pyn. That is, the specifications of  the given logics were preserved.\\[1mm]
(2)  Observe that the formal specification for the first coordinate of the multioperator $\tilde{\to}$ associated to $\to$, the only connective being shared, coincide in the swap structures for \pyn\ and \Tm: it is given by the term $t(z,w)=z_1 \Rightarrow w_1$. This allows to perform the combination of both multialgebras in this signature.\\[1mm] 
(3) Despite the only shared conective is $\to$, the definition of the multialgebra for \INT\ takes into account the multioperations $\hat{\land}$ and $\hat{\vee}$ defined in \Tm\ (recall Remark~\ref{conj-disj-Tm}). Hence, \INT\ should be seen as being obtained by combining \pyn\ with the expansion of \Tm\ by  $\land$ and $\vee$, and by sharing $\{\land,\vee,\to\}$, which are shareable (the first coordinate of the interpretation of $\#$ in both logics is $z_1 \# w_1$ for $\#\in\{\land,\vee,\to\}$). This justifies the definition of $\tilde{\land}$ and $\tilde{\vee}$ in \INT. \\[1mm]
(4) The restrictions on the 4-tuples in $\textsc{B}_{T}$ will be represented by some axioms and theorems of the Hilbert calculus for \INT\ to be defined in Section~\ref{Hilbert}. Thus, $z_2 \leq z_1$ and $z_1 \sqcap z_3=0$ will correspond to axiom \axT\ of Definition~\ref{sysB4}. In turn,  items~(1) and~(5) of Proposition~\ref{propbaslog} will correspond to the restrictions $z_2 \sqcap z_4=0$ and $z_3 \leq z_4$, respectively (see Section~\ref{Hilbert}).
It is worth noting that the conditions over the second coordinate of the multioperation $z \tilde{\to} w$ will correspond to axioms \axK, \axKu\ and \axMu\ of the Hilbert calculus for \INT, to be introduced in Definition~\ref{IvNelLHil} below.
\end{remarks}

\noindent The (non-deterministic) truth-tables of the Nmatrices for the logics above can be displayed as follows:

{\scriptsize
\begin{center}
\begin{tabular}{|c|c|c|c|c|c|c|}
\hline
 $\to$ & $T_0$  & $t_0$  & $t_1$ & $f_0$ & $f_1$  & $F_1$ \\[1mm]
 \hline 
    $T_0$     & $T_0$ & $t_0$ & $t_1$ & $f_0$ & $f_1$ & $F_1$    \\[1mm] \hline
     $t_0$    & $T_0$  & $T_0, t_0$ & $t_1$ & $f_0$ & $f_1$ & $f_1$  \\[1mm] \hline
     $t_1$    & $T_0$  & $T_0, t_0$ & $t_1$ & $f_0$ & $f_1$ & $f_1$  \\[1mm] \hline
     $f_0$    & $T_0$  & $T_0, t_0$ & $T_0, t_0$ & $T_0, t_0$ & $T_0, t_0$ & $t_0$  \\[1mm] \hline
     $f_1$    & $T_0$  & $T_0, t_0$ & $T_0, t_0$ & $T_0, t_0$ & $T_0, t_0$ & $t_0$  \\[1mm] \hline
     $F_1$    & $T_0$  & $T_0$ & $T_0$ & $T_0$ & $T_0$ & $T_0$  \\[1mm] \hline
\end{tabular}
\hspace{0.3cm}
\begin{tabular}{|c|c|c|c|} \hline
$\quad$ & $\sneg$  & $\neg$  & $\square$\\[1mm]
 \hline
    $T_0$   & $F_1$   & $F_1$    & \textrm{D}  \\[1mm] \hline
     $t_0$   & $f_0, f_1$   & $f_1$   & \textrm{ND}\\[1mm] \hline
     $t_1$   &$f_0, f_1$  &$t_1$   & \textrm{ND} \\[1mm] \hline
     $f_0$   & $t_0, t_1$    & $f_0$   & \textrm{ND} \\[1mm] \hline
     $f_1$   & $t_0, t_1$   & $t_0$   & \textrm{ND} \\[1mm] \hline
     $F_1$   & $T_0$   & $T_0$   & \textrm{ND} \\[1mm] \hline
\end{tabular}

\end{center}
}

{\scriptsize
\noindent
\begin{center}
\begin{tabular}{|c|c|c|c|c|c|c|}
\hline
$\land$ & $T_0$  & $t_0$  & $t_1$ & $f_0$ & $f_1$  & $F_1$ \\[1mm]
 \hline
    $T_0$     & $T_0$ & $t_0$ & $t_1$ & $f_0$ & $f_1$ & $F_1$    \\[1mm] \hline
     $t_0$    & $t_0$  & $t_0$ & $t_1$ & $f_0$ & $F_1, f_1$ & $F_1$  \\[1mm] \hline
     $t_1$    & $t_1$  & $t_1$ & $t_1$ & $F_1, f_1$ & $F_1, f_1$ & $F_1$  \\[1mm] \hline
     $f_0$    & $f_0$  & $f_0$ & $F_1, f_1$ & $f_0$ & $F_1, f_1$ & $F_1$  \\[1mm] \hline
     $f_1$    & $f_1$  & $F_1, f_1$ & $F_1, f_1$ & $F_1, f_1$ & $F_1, f_1$ & $F_1$  \\[1mm] \hline
     $F_1$    & $F_1$  & $F_1$ & $F_1$ & $F_1$ & $F_1$ & $F_1$  \\[1mm] \hline
\end{tabular}

\end{center}
}

{\scriptsize
\noindent
\begin{center}
\begin{tabular}{|c|c|c|c|c|c|c|}
\hline
$\lor$ & $T_0$  & $t_0$  & $t_1$ & $f_0$ & $f_1$  & $F_1$ \\[1mm]
 \hline 
    $T_0$     & $T_0$ & $T_0$ & $T_0$ & $T_0$ & $T_0$ & $T_0$    \\[1mm] \hline
     $t_0$    & $T_0$  & $T_0, t_0$ & $T_0, t_0$ & $T_0, t_0$ & $T_0, t_0$ & $t_0$  \\[1mm] \hline
     $t_1$    & $T_0$  & $T_0, t_0$ & $t_1$ & $T_0, t_0$ & $t_1$ & $t_1$  \\[1mm] \hline
     $f_0$    & $T_0$  & $T_0, t_0$ & $T_0, t_0$ & $f_0$ & $f_0$ & $f_0$  \\[1mm] \hline
     $f_1$    & $T_0$  & $T_0, t_0$ & $t_1$ & $f_0$ & $f_1$ & $f_1$  \\[1mm] \hline
     $F_1$    & $T_0$  & $t_0$ & $t_1$ & $f_0$ & $f_1$ & $F_1$  \\[1mm] \hline
\end{tabular}
\end{center}
}

\

\begin{remarks} \label{Box-not-congr} \ \\
(1) Observe that $\square(\varphi \to \psi)$ is not equivalent in general  to $\square(\sneg\varphi \vee \psi)$ in either \INT\ or \Tm. To see this, consider two different propositional variables $p$ and $q$, and let $v$ be a valuation over the 6-valued Nmatrix for \INT\ such that $v(p)=f_1$ , $v(q)=t_1$ and $v(\sneg p)=t_0 \in \sneg f_1$. Then, $v(p \to q) \in f_1 \to t_1 = \{T_0,t_0\}$. From this, $v(\sneg p \vee q) \in  t_0 \vee t_1 = \{T_0,t_0\}$. Hence, by taking $v$ such that $v(p \to q)=T_0$ and  $v(\sneg p \vee q)=t_0$ we have that $v(\square(p \to q)) \in \textrm{D}$ while $v(\square(\sneg p \vee q)) \notin \textrm{D}$. The same property holds for the 4-valued Ivlev's system \Tm. Indeed, it is very simple to adapt the previous  counterexample, by replacing $t_0$ and $t_1$ by $t$; $T_0$ by $T$ and $f_1$ by $f$, and by observing that $\sneg p \vee q=\sneg\sneg p \to q$ in \Tm\ (recall Remark~\ref{conj-disj-Tm}). Then, we have that both $v(\sneg p \vee q)=v(\sneg\sneg p \to q)$ and $v(p \to q)$ belong to $f \to t = \{T,t\}$. Hence, it is enough taking $v(\sneg p \vee q)=t$  and $v(p \to q)=T$, obtaining so that $v(\square(p \to q)) \in \textrm{D}$ while $v(\square(\sneg p \vee q)) \notin \textrm{D}$.  \\[1mm]
(2) In a similar way, it is easy to show that $\square(\sneg \varphi \vee \sneg \psi)$ is not equivalent in general  to $\square \sneg(\varphi \land \psi)$ in either \INT\ or \Tm. Indeed, in \Tm, this corresponds to the equivalence between $\square(\sneg\sneg \varphi \to \sneg \psi)$ and $\square (\varphi  \to \sneg  \psi)$ which, as observed above, does not hold in general. In the 6-valued system \INT, it is enough considering $\varphi=p$ and $\psi=q$ where $p$ and $q$ are two different propositional variables, and a valuation $v$ such that $v(p)=f_0$, $v(q)=f_1$, $v(\sneg \varphi \vee \sneg \psi)=T_0$ and $v(\sneg(\varphi \land \psi))=t_0$.\\[1mm]
(3) It can be proven analogously that $\square(\neg \varphi \vee \neg \psi)$ is not equivalent in general to $\square \neg(\varphi \land \psi)$ in \INT. To see this, consider $\varphi=\psi=p$ for  a propositional variable $p$, and let $v$ be a valuation such that $v(p)=f_1$, $v(\neg \varphi \vee \neg \psi)=T_0$ and $v(\neg(\varphi \land \psi))=t_0$. This is always possible, since $\neg f_1 \vee \neg f_1 = t_0 \vee t_0=\{T_0,t_0\}$ and so we can choose $v(\neg \varphi \vee \neg \psi)=T_0$. In turn, $f_1 \land f_1=\{F_1,f_1\}$. Hence, we can take $v(p \land p)=f_1$  and so $v(\neg(\varphi \land \psi))\in \neg f_1=\{t_0\}$. From this, $v(\square(\neg \varphi \vee \neg \psi)) \in \textrm{D}$ while $v(\square \neg(\varphi \land \psi))  \notin \textrm{D}$.\\[1mm]
(4) Note that $\neg(\varphi \to \psi)$ is equivalent to $\varphi \land \neg \psi$ in \pyn\ and \INT. However, in general $\neg\neg(\varphi \to \psi)$ is not equivalent to $\neg(\varphi \land \neg \psi)$ in either \pyn\ or \INT. Indeed, let $\varphi=\psi=p$ for a propositional variable $p$, and let $v$ be  valuation over the 4-valued matrix for  \pyn\ such that $v(p)=\nei$. Then, $v(\neg\neg(p \to p))={\bf 1} \in D_4$, but $v(\neg(p \land \neg p))=\nei \notin D_4$. For \INT\ take a valuation $v'$ over the 6-valued Nmatrix for \INT\ such that $v'(p)=f_0$. Then, $v'(\neg\neg(p \to p)) \in \{T_0, t_0\} \subseteq D$, but $v'(\neg(p \land \neg p))=f_0 \notin D$.

The examples above show that neither \INT\ nor \Tm\ nor \pyn\ are self-extensional, given that $\square$ and $\neg$ do not preserve logical equivalences. That is, $\square$ and $\neg$ are hyperintensional operators.
\end{remarks}

\

\noindent As mentioned in the introduction, there is a close relationship between the combined logic \INT\ (as well as its axiomatic extensions to be analyzed in Section~\ref{other-logics}) and the logics and evidence and truth (\lets) originally proposed in~\cite{kolkata} (see also~\cite{car:rod:19}). Indeed, besides \INT\ being an expansion of \fde, it is possible to define, in a natural way, a classicality connective $\circ$ w.r.t. $\neg$ as follows: $\cons\varphi \defi (\varphi \lor \neg\varphi) \land \sneg(\varphi \land \neg\varphi)$. This connective is interpreted by the following multioperator in $\mathcal{M}_{\INT}$:

\

{\scriptsize
\begin{center}
\begin{tabular}{|c|c|c|c|c|c|c|} \hline
$\varphi$ & $\sneg\varphi$  & $\neg\varphi$  & $\varphi \lor \neg\varphi$  & $\varphi \land \neg\varphi$  & $\sneg(\varphi \land \neg\varphi)$  & $\cons\varphi \defi (\varphi \lor \neg\varphi) \land \sneg(\varphi \land \neg\varphi)$ \\[1mm]
 \hline 
    $T_0$   & $F_1$   & $F_1$    & $T_0$  & $F_1$ & $T_0$ & $T_0$ \\[1mm] \hline
     $t_0$   & $f_0, f_1$   & $f_1$   & $T_0, t_0$ & $F_1, f_1$ & \textrm{D} & \textrm{D}  \\[1mm] \hline
     $t_1$   &$f_0, f_1$  & $t_1$   & $t_1$  & $t_1$ & $f_0, f_1$  & $F_1, f_1$    \\[1mm] \hline
     $f_0$   & $t_0, t_1$    & $f_0$   & $f_0$ & $f_0$ & $t_0, t_1$ & \textrm{ND}  \\[1mm] \hline
     $f_1$   & $t_0, t_1$   & $t_0$   & $T_0, t_0$ & $F_1, f_1$ & \textrm{D} & \textrm{D}  \\[1mm] \hline
     $F_1$   & $T_0$   & $T_0$   & $T_0$  & $F_1$ & $T_0$ & $T_0$ \\[1mm] \hline
\end{tabular}
\end{center}}

\

\noindent
Clearly, $\circ$ simultaneously recovers the explosion law and the excluded middle:  $\cons\varphi$ derives $\varphi \vee\neg\varphi$, and from $\varphi \land \neg\varphi\land\cons\varphi$ everything follows. Thus \INT, as well as the axiomatic extensions to be studied in Section~\ref{other-logics}, are modal \lets. Observe that, for every valuation $v$ over  $\mathcal{M}_{\INT}$,
$v(\cons\varphi) \in \textrm{D}$ iff $v(\varphi) \notin\{t_1,f_0\}$. This means that $t_1$ and $f_0$ are the only paradefinite truth-values. As we shall see in Section~\ref{conserva}, they correspond respectively to \bo\ and \nei, and they describe the only doubtful informational states, as $\circ$ indicates.

\section{Hilbert calculus for \INT} \label{Hilbert}

In this section, a sound and complete Hilbert calculus for \INT\ will be presented. It will be obtained by combining the corresponding Hilbert calculi for \pyn\ and for \Tm, together with some suitable bridge axioms which guarantee the coherence with the proposed semantical approach. Let us start by considering Hilbert calculi for \pyn\ and \Tm.

A Hilbert calculus for \pyn\ can be obtained  by adding to the axiomatization of $\nel$ given in~\cite{odin:05} the Peirce law $\varphi \vee (\varphi \to \psi)$. That is:

\

\begin{definition} [Logic \pyn] \label{sysB4} The logic \pyn\ is presented in the signature $\Sigma_4=\{\land,\lor,\to,\neg\}$ through the following Hilbert  calculus (where $\varphi \sse \psi$ is an abbreviation for $(\varphi \to \psi) \land (\psi \to \varphi)$):\\[2mm]
{\bf Axiom schemas:}\\[3mm]
	$\begin{array}{ll}
	\axu & \varphi \to (\psi \to \varphi)\\[2mm]
	\axd & (\varphi \to (\psi \to \gamma)) \to ((\varphi \to \psi) \to (\varphi \to \gamma))\\[2mm]
	\axt & \varphi \to (\psi \to (\varphi \land \psi))\\[2mm]
	\axq & (\varphi \land \psi) \to \varphi	\\[2mm]
	\axc & (\varphi \land \psi) \to \psi\\[2mm]
	\axs &  \varphi \to (\varphi \lor \psi)\\[2mm]
	\axst & \psi \to (\varphi \lor \psi)\\[2mm]
	\axo & (\varphi \to \gamma) \to ((\psi \to \gamma) \to	((\varphi \lor \psi) \to \gamma))\\[3mm]
\end{array}
$
\noindent	$\begin{array}{ll}
	\axDN & \neg\neg\varphi \sse \varphi\\[2mm]
	\axdmu & \neg(\varphi \lor \psi) \sse (\neg \varphi \land \neg \psi)\\[2mm]
	\axdmd & \neg(\varphi \land \psi) \sse (\neg \varphi \lor \neg \psi)  \\[2mm]
	\axdmt & \neg(\varphi \to \psi) \sse (\varphi \land \neg \psi) \\[2mm]
\axouimp	& \varphi \lor (\varphi \to \psi)
\end{array}
$

\noindent	$\begin{array}{ll}
	\mbox{{\bf Inference rule:}}\\[2mm]
	
	\MP: & \displaystyle\frac{\varphi \ \ \ \ \varphi\to	\psi}{\psi}
\end{array}
$
\end{definition}

\

\noindent For every $\Gamma \cup \{\varphi\} \subseteq For(\Sigma_4)$ we write $\Gamma \vdash_{\pyn}\varphi$ to denote that $\varphi$ is derivable from $\Gamma$ in the Hilbert calculus for \pyn.

\

\begin{remark} 
As already mentioned in the Introduction, the extension of \pyn\ by adding a bottom formula $\bot$ such that $\neg\bot$ is valid  was proposed by Odintsov, in the context of the algebraic study of \nel,   under the name of \bfour. In the axiomatic formulation of \bfour\ given in~\cite{odin:wan:17}, the Peirce's law \axouimp\ is presented  as $((\varphi \to \psi) \to \varphi) \to \varphi$, but that formulation is equivalent to the present one. Observe that \axu-\axo\ plus \MP\ corresponds to positive intuitionistic propositional  logic \iplp, hence by adding \axouimp\ we get positive classical propositional  logic \cplp. By removing from the system above axiom  \axouimp\ it is obtained a Hilbert calculus for \nel.
\end{remark}

\

\noindent
Now, a Hilbert calculus for Ivlev's logic \Tm\ will be presented. By convenience, the original Ivlev's axiomatization will be slightly modified.

\

\begin{definition} [Ivlev 4-valued modal logic \Tm] \label{ivev-sys} The logic \Tm\ is defined by the following Hilbert calculus over  the signature  $\Sigma_m=\{\to,\sneg,\square\}$ (where the disjunction and conjunction are defined in terms of $\sneg$ and $\to$ as  $\varphi \vee \psi \defi \sneg\varphi \to \psi$ and  $\varphi \land \psi \defi \sneg(\varphi \to \sneg\psi)$, respectively. Hence, $\varphi \sse \psi$ is an abbreviation for $\sneg((\varphi \to \psi) \to \sneg(\psi \to \varphi))$):\\[2mm]
{\bf Axiom schemas:}\\[3mm]
$\begin{array}{ll}
	\axu & \varphi \to (\psi \to \varphi)\\[2mm]
	\axd & (\varphi \to (\psi \to \gamma)) \to ((\varphi \to \psi) \to (\varphi \to \gamma))\\[2mm]
	\axn & (\sneg \psi \to \sneg \varphi) \to ((\sneg \psi \to \varphi) \to \psi)\\[2mm]
	\axK & \square (\varphi \to \psi) \to (\square \varphi \to \square \psi)\\[2mm]
	\axKu & \square (\varphi \to \psi) \to (\square \sneg\psi \to \square \sneg \varphi))\\[3mm]
\end{array}$
$\begin{array}{ll}
	\axKd &  \square \sneg(\varphi \to \psi) \sse (\square \varphi \land \square \sneg \psi)\\[2mm]
	\axMu & (\square \sneg \varphi \vee \square \psi) \to \square(\varphi \to \psi)\\[2mm]
	\axT & \square \varphi \to \varphi\\[2mm]
	\axDNu & \square \varphi \sse \square \sneg \sneg \varphi
\end{array}$	

\noindent	$\begin{array}{ll}
	\mbox{{\bf Inference rule:}}\\[2mm]
	
	\MP: & \displaystyle\frac{\varphi \ \ \ \ \varphi\to	\psi}{\psi}\\[2mm]
\end{array}
$

\

\noindent For every $\Gamma \cup \{\varphi\} \subseteq For(\Sigma_m)$ we write $\Gamma \vdash_{\Tm}\varphi$ to denote that $\varphi$ is derivable from $\Gamma$ in the Hilbert calculus for \Tm.
\end{definition}

\noindent
Observe that the Hilbert calculus given by \axu, \axd\ and \axn\ plus \MP\ characterizes classical propositional  logic \cpl\ over the signature $\{\to,\sneg\}$, therefore \cpl\ is contained in \Tm. From this, the (defined) conjunction and disjunction connectives behave exactly as in classical logic, given that \Tm\ conservatively expands \cpl.

Now we are ready to introduce axiomatically the paradefinite modal systems we propose, by combining simultaneously the features of the Ivlev-like modal systems of Definition~\ref{ivev-sys} with the paradefinite system \pyn\ introduced in Definition~\ref{sysB4}. Since $\land$ and $\lor$ will be primitive connectives  (because of its \pyn-fragment), some interaction axioms must be added, provided that the operator  $\square$ is not congruential, that is, it does not preserve logical equivalences in general.  Six specific axioms involving the interaction between the two negations $\sneg$, $\neg$ and the modal operator $\square$ (axioms \axNDN-\axNde\ below) must be included, in order to satisfy the conceptual requirements about the nature of the truth values of the combined system mentioned in Section~\ref{sectINT}.
See the comments in Remark~\ref{obs-defisys} below.

\begin{definition} [Paradefinite Ivlev-like modal logic \INT\ over \fde] \label{IvNelLHil}
The  system of paradefinite Ivlev-like modal logic \INT\ is defined over the signature $\Sigma=\{\land,\lor,\to,\sneg,\neg,\square\}$  as the Hilbert calculus obtained from the one for \pyn\ (recall Definition~\ref{sysB4}) by	adding all the  axiom schemas from the Hilbert calculus for \Tm\ (recall Definition~\ref{ivev-sys},  taking into account that $\land$ in  \axKd\ and $\vee$ in \axMu\ denote now the primitive connectives for conjunction and disjunction of $\Sigma$ instead of being abbreviations, and that the abbreviation for $\sse$ in axioms \axKd\ and \axDNu\ is the one adopted in  Definition~\ref{sysB4} instead of the one  considered in Definition~\ref{ivev-sys}), plus the following axiom schemas (in which $\sse$ is as in  Definition~\ref{sysB4}):

$\begin{array}{ll}
\axC	& \square(\varphi \land \psi) \sse (\square\varphi \land \square\psi) \\[2mm]
\axnD	& \square \sneg(\varphi \vee \psi) \sse  \square(\sneg\varphi \land \sneg\psi) \\[2mm]
\axnC	&  (\square\sneg\varphi \vee \square\sneg\psi) \to \square \sneg(\varphi \land \psi) \\[2mm]
\axD	& (\square\varphi \vee \square\psi) \to \square(\varphi \lor \psi) \\[2mm]
\axNDN	 & \square \varphi \sse \square \sneg \neg \varphi
\end{array}$
$\begin{array}{ll}
\axNTN	 & \square \sneg\varphi \sse \square \sneg \neg\neg \varphi\\[2mm]
\axNN	 & \square \sneg(\varphi \land \psi) \to (\square \varphi \to \square\sneg\psi)\\[2mm]
\axNT	 & \square \sneg(\varphi \land \psi) \to (\square \psi \to \square\sneg\varphi)\\[2mm]
\axNo	 & \square (\varphi \lor \psi) \to (\square \sneg\varphi \to \square\psi)\\[2mm]
\axNde	 & \square (\varphi \lor \psi) \to (\square \sneg\psi \to \square\varphi)\\
\end{array}
$
\end{definition}

\

\begin{proposition} \label{propbaslog} The following formulas are theorems of \INT:\\[1mm]
$\begin{array}{ll}
(1) & \square \varphi \to \sneg\neg \varphi.\\[1mm]
(2) & \square \varphi \sse \square \neg \sneg \varphi.\\[1mm]
(3) & \square \varphi \sse \square \neg \neg \varphi.\\[1mm]
(4) & \square \neg\varphi \sse \square\sneg \varphi.\\[1mm]
(5) & \square\sneg\varphi \to \neg\varphi.
\end{array}$
\hspace{1cm}$\begin{array}{ll}
(6) & \square\neg(\varphi \lor\psi) \sse \square(\neg\varphi \land \neg\psi). \\[1mm]
(7) & \square\neg(\varphi \to\psi) \sse \square(\varphi \land \neg\psi). \\[1mm]
(8) & (\square\neg\varphi \vee \square\neg\psi) \to \square \neg(\varphi \land \psi).\\[1mm]
(9) & \square\varphi \to \neg\sneg \varphi.\\[1mm]
(10) & \neg\sneg \varphi \to \sneg\square\sneg \varphi.
\end{array}$
\end{proposition}
\begin{proof} Along this proof, we will make implicit use of the deduction metatheorem in \INT. Moreover,  some obvious steps in the derivations in \INT\ presented below will be omitted, for the sake of simplicity.\\
(1) Consider the following derivation in \INT:\\
\indent 1. $\square \varphi$ \ (Hyp)\\
\indent 2. $\square \varphi \to \square \sneg \neg \varphi$ \ \axNDN\\
\indent 3.  $\square \sneg \neg \varphi$ \ (1, 2, \MP)\\
\indent 4.  $\square \sneg \neg \varphi \to \sneg \neg \varphi$ \ \axT\\
\indent 5. $\sneg \neg \varphi$  \ (3, 4, \MP)\\[1mm]
(2)  $\square \varphi \equiv_{\axDNu} \square \sneg \sneg \varphi \equiv_{\axNTN} \square \sneg \neg \neg \sneg \varphi  \equiv_{\axNDN}  \square \neg \sneg \varphi$.\\[1mm]
(3) $\square \varphi \equiv_{\axNDN} \square \sneg \neg \varphi  \equiv_{\axNTN} \square \sneg \neg\neg\neg \varphi \equiv_{\axNDN} \square \neg\neg\varphi$.\\[1mm]
(4) $\square \neg \varphi \equiv_{\axNDN} \square \sneg \neg\neg \varphi \equiv_{\axNTN} \square \sneg \varphi$.\\[1mm]
(5)  Consider the following derivation in \INT:\\
\indent 1. $\square \sneg\varphi$ \ (Hyp)\\
\indent 2. $\square \sneg \varphi \to \square \neg \varphi$ \ (by (4))\\
\indent 3.  $\square \neg \varphi$ \ (1, 2, \MP)\\
\indent 4.  $\square \neg \varphi \to \neg \varphi$ \ \axT\\
\indent 5. $\neg \varphi$  \ (3, 4, \MP)\\[1mm]
(6) $\square\neg(\varphi \lor\psi) \equiv_{(4)} \square\sneg(\varphi \lor\psi) \equiv_{\axnD,\axC} \square\sneg \varphi \land \square\sneg\psi \equiv_{(4)} \square\neg \varphi \land \square\neg\psi \equiv_{\axC} \square(\neg\varphi \land \neg\psi)$. \\[1mm]
(7) $\square\neg(\varphi \to\psi) \equiv_{(4)} \square\sneg(\varphi \to\psi) \equiv_{\axKd} \square\ \varphi \land \square\sneg\psi \equiv_{(4)} \square \varphi \land \square\neg\psi \equiv_{\axC} \square(\varphi \land \neg\psi)$. \\[1mm]
(8)  Consider the following derivation in \INT:\\
\indent 1. $(\square\neg\varphi \vee \square\neg\psi)$ \ (Hyp)\\
\indent 2. $(\square\neg\varphi \vee \square\neg\psi) \to (\square\sneg\varphi \vee \square\sneg\psi)$ \ (by (4))\\
\indent 3.  $(\square\sneg\varphi \vee \square\sneg\psi)$ \ (1, 2, \MP)\\
\indent 4.  $(\square\sneg\varphi \vee \square\sneg\psi) \to \square \sneg(\varphi \land \psi)$ \ \axnC\\
\indent 5. $\square \sneg(\varphi \land \psi)$  \ (3, 4, \MP)\\
\indent 6. $\square \sneg(\varphi \land \psi) \to \square \neg(\varphi \land \psi)$ \ (by (4))\\
\indent 7.  $\square \neg(\varphi \land \psi)$ \ (5, 6, \MP)\\[1mm]
(9)  Consider the following derivation in \INT:\\
\indent 1. $\square \varphi$ \ (Hyp)\\
\indent 2. $\square \varphi \to \square  \sneg\sneg \varphi$ \ \axDNu\\
\indent 3.  $\square  \sneg\sneg \varphi$ \ (1, 2, \MP)\\
\indent 4.  $\square  \sneg\sneg\varphi \to \neg\sneg \varphi$ \ (by (5))\\
\indent 5. $\neg\sneg \varphi$  \ (3, 4, \MP)\\[1mm]
(10)  Consider the following derivation in \INT:\\
\indent 1. $\square \sneg\varphi \to \sneg\neg\sneg \varphi$ \ (by (1))\\
\indent 2. $(\square \sneg\varphi \to \sneg\neg\sneg \varphi) \to (\neg\sneg \varphi \to \sneg\square \sneg\varphi)$ \ (by \cplp)\\
\indent 3.  $\neg\sneg \varphi \to \sneg\square \sneg\varphi$ \ (1, 2, \MP).
\end{proof}

\begin{remarks} \label{obs-defisys} \ \\
(1) Recall that Ivlev defined the disjunction $\vee$ in terms of $\sneg$ and $\to$ as it is usually done in classical logic, namely $\varphi \vee \psi \defi \sneg\varphi \to \psi$. Hence, by replacing $\varphi$ by $\sneg\varphi$ in \axMu, and by \axDNu, we get the theorem $(\square \varphi \vee \square \psi) \to \square(\varphi \vee \psi)$ in the calculus for \Tm\  (using Ivlev's  abbreviations). This shows that the role of  axiom \axMu\ in Ivlev's approach is (at least, see item~(5) below) to validate this property of $\square$ w.r.t. the defined disjunction. Since disjunction is a primitive connective in \INT, it is necessary to add explicitly axiom \axD\ in order to guarantee this basic property of $\square$ w.r.t. the primitive disjunction. Indeed, despite $\varphi \vee \psi$ being equivalent in this system to $\sneg\varphi \to \psi$ (given that it contains classical logic), it is not true in general that  $\square(\varphi \vee \psi)$ is equivalent to $\square(\sneg\varphi \to \psi)$. The reason is that $\square$ does not preserves logical equivalences in general (see Remarks~\ref{Box-not-congr} above). \\
(2) Analogously, conjunction is defined in Ivlev's \Tm\ as in classical logic, namely  $\varphi \land \psi \defi \sneg(\varphi \to \sneg\psi)$. By replacing $\psi$ by $\sneg\psi$ in \axKd, and by \axDNu, we get in \Tm\ the theorem $(\square \varphi \land \square \psi) \sse \square(\varphi \land \psi)$. This shows that the aim of Ivlev's axiom  \axKd\  is (at least, see item~(4) below)) to validate the distributivity property of $\square$ w.r.t. conjunction, besides expressing a desirable logical property by itself. Given that conjunction is a primitive connective in \INT, the role of axiom \axC\ is to guarantee this basic property of $\square$, while keeping the original Ivlev's axiom \axKd. Observe once again that, despite $\varphi \land \psi$ being equivalent in this system to $\sneg(\varphi \to \sneg\psi)$,  $\square$ does not preserve logical equivalences in general, hence axiom \axC\ is necessary.\\
(3) Still concerning axiom \axKd\, if we substitute in it $\varphi$ by $\sneg\varphi$ we obtain, in \Tm's notation, the theorem $\square \sneg(\varphi \vee \psi) \sse (\square \sneg\varphi \land \square \sneg \psi)$. So, the aim of axiom \axnD\ is to preserve this property when $\vee$ is primitive, while keeping the original Ivlev's presentation  (observe that the conjunction on the right-hand side of  \axKd\ occurs outside the scope of $\square$, hence it is irrelevant whether this occurrence of $\land$ is the definite or the primitive conjunction).\\
(4) By replacing $\psi$ by $\sneg\psi$ in \axMu, and by \axDNu, we get in system \Tm\ the theorem $(\square \sneg\varphi \vee \square \sneg\psi) \to \square\sneg(\varphi \land \psi)$ (by using \Tm's notation). Hence, the role played by \axnC\  is to keep this property when $\land$ is a primitive connective (observe that the disjunction in  \axMu\ occurs out of the scope of $\square$, hence it is irrelevant whether it is the defined or the primitive $\vee$).\\
(5) In the same vein, axioms \axNN-\axNde\ are clearly derivable in the Hilbert calculus for \Tm\ (when considering the corresponding  abbreviations for $\land$, $\vee$ and $\sse$ in the signature $\Sigma_m$). Since these theorems of \Tm,  expressed in the signature of \INT, are essential for the proof of the Truth Lemma (Proposition~\ref{TruthL} below), they were included in the Hilbert calculus for \INT.
\end{remarks}

\section{Soundness and completeness}

In this section  it will be shown that the Hilbert calculus for \INT\ is sound and complete w.r.t. the Nmatrix $\mathcal{M}_{\INT}$.

\begin{theorem} [Soundness] \label{sound-INT}
Let $\Gamma \cup \{\varphi\}$ be a set of formulas. Then, $\Gamma \vdash_{\INT} \varphi$ implies that $\Gamma \models_{\INT} \varphi$.
\end{theorem}
\begin{proof} Suppose that $\Gamma \vdash_{\INT} \varphi$.
It is enough to prove that, if $\delta$ is an instance of an axiom of \INT\ and $v$ is a valuation over $\mathcal{M}_{\INT}$ then $v(\delta)_1=1$, where $v(\delta)_1$ denotes the first coordinate of $v(\delta)$. Indeed, by proving this, and by the fact that $v(\delta) \in  \textrm{D}$ and  $v(\delta \to \psi) \in  \textrm{D}$ implies that $v(\psi) \in  \textrm{D}$, the result follows  by induction on the length of any derivation of $\varphi$ from $\Gamma$ in the Hilbert calculus for \INT. Thus, let  $\delta$ be an instance of an axiom of \INT\ and let $v$ be a valuation over $\mathcal{M}_{\INT}$. If $\delta$ is an instance of one of the axioms \axu-\axn\ or \axouimp\ of \cpl\ then clearly $v(\delta)_1=1$, by definition of the multioperations in $\mathcal{M}_{\INT}$: it is enough to observe that, in $v(\delta)_1$, every connective of $\delta$ in $\{\land, \vee,\to, \sneg\}$ is interpreted as the corresponding Boolean operator in {\bf 2}. For instance, $v((\gamma \land \psi) \to \psi)_1=(v(\gamma)_1 \sqcap v(\psi)_1) \Rightarrow v(\psi)_1$, which is always $1$ in {\bf 2} for any value  $v(\gamma)$ and $v(\psi)$. Concerning axioms \axDN, \axdmu, \axdmd\ and \axdmt\ inherited from \nel, they are also valid because of the definition of the multioperators in $\mathcal{M}_{\INT}$: it follows by an easy adaptation of the fact that the Hilbert calculus of \nel\ is sound w.r.t. twist structures semantics over implicative lattices (see, for instance, \cite[Theorem~2.2]{odin:05}). Now, suppose that $\delta$ is a modal axiom (that is, an axiom for $\square$). It can be proven that $v(\delta)_1=1$ because of the restrictions on the snapshots in $\textsc{B}_{T}$, by the definition of the multioperators in  $\mathcal{M}_{\INT}$, and by the following facts: $v(\psi \# \gamma)_1=v(\psi)_1 \# v(\gamma)_1$ ($\# \in \{\land, \vee, \to\}$) and $v(\sneg\psi)_1=\sneg v(\psi)_1$ (where, by abuse of notation, the symbols used for the classical connectives and for  the Boolean operators interpreting them are the same); $v(\square\psi)_1=v(\psi)_2$; and $v(\square\sneg\psi)=v(\psi)_3$. For instance, let $\delta=\square(\psi \to \gamma) \to (\square \psi \to \square\gamma)$ be an instance of axiom \axK. Then, since $v(\psi \to \gamma) \in v(\psi) \tilde{\to} v(\gamma)$, $v(\square(\psi \to \gamma))_1=v(\psi \to \gamma)_2 \leq v(\psi)_2 \Rightarrow v(\gamma)_2=v(\square\psi)_1\Rightarrow v(\square\gamma)_1=v(\square \psi \to \square\gamma)_1$, hence $v(\delta)_1=1$. By analogous reasoning, $v(\square(\psi \to \gamma))_1 \leq v(\square \sneg\psi \to \square \sneg \varphi)_1$, which proves that every instance of \axKu\ is valid. In turn, $v(\square\psi)_1=v(\psi)_2 \leq v(\psi)_1$, by the fact that $v(\psi) \in \textsc{B}_{T}$, proving that every instance of \axT\ is valid. Concerning \axNTN, observe that $v(\square\sneg\neg\neg\psi)_1=v(\neg\neg\psi)_3=v(\neg\psi)_2=v(\psi)_3=v(\square\sneg\psi)_1$. By a similar reasoning it can be proven that every modal axiom of \INT\ is valid in $\mathcal{M}_{\INT}$. This completes the proof.
\end{proof}

\noindent
In order to prove completeness, some  standard  notions and results from (Tarskian) abstract logic  need to be used.
Recall that, given a Tarskian and finitary logic {\bf L}\footnote{A logic {\bf L} with consequence relation $\vdash_{\bf L}$ is {\em Tarskian} if it satisfies the following:~(1)~if $\alpha \in \Gamma$ then $\Gamma \vdash_{\bf L} \alpha$; (2)~if $\Gamma \vdash_{\bf L} \alpha$ and $\Gamma \subseteq \Delta$ then $\Delta \vdash_{\bf L} \alpha$; and~(3)~if $\Delta \vdash_{\bf L} \alpha$ and $\Gamma \vdash_{\bf L} \beta$ for every $\beta \in \Delta$ then $\Gamma \vdash_{\bf L} \alpha$. A Tarskian logic {\bf L} is  {\em finitary} if it satisfies: (4)~if $\Gamma \vdash_{\bf L} \alpha$ then there exists  a finite subset $\Gamma_0$ of $\Gamma$ such that $\Gamma_0 \vdash_{\bf L} \alpha$.} and a set of formulas $\Delta \cup \{ \varphi\}$ of {\bf L}, $\Delta$ is {\em $\varphi$-saturated in} {\bf L} if:~(1)~$\Delta \nvdash_{\bf L} \varphi$; and~(2)~if $\psi \notin \Delta$ then $\Delta,\psi \vdash_{\bf L}\varphi$.

\begin{remarks} \ \\\label{saturated-rem}
(1) It is easy to prove that, if $\Delta$ is $\varphi$-saturated in a  Tarskian logic {\bf L} then it is a closed theory in {\bf L}, that is: $\psi \in \Delta$ iff $\Delta \vdash_{\bf L} \psi$.\\
(2) By a well-known result due to Lindenbaum and \L o\'s,  if $\Gamma \cup \{ \varphi\}$ is a set of formulas of a Tarskian and finitary logic {\bf L} such that $\Gamma \nvdash_{\bf L} \varphi$, there exists a $\varphi$-saturated set $\Delta$  such that $\Gamma \subseteq \Delta$. A proof of this fact can be found, for instance, in~\cite[Theorem~22.2]{woj:84} or in~\cite[Theorem~2.2.6]{car:con:16}. Clearly, the logic associated to the Hilbert calculus of \INT\ is Tarskian and finitary, hence the result by Lindenbaum and \L o\'s mentioned above hold for it.
\end{remarks}

\begin{proposition} \label{equiv-compo} In the Hilbert calculus for \INT\  the following holds: \\ 
(1) $\theta_{T_0}(\varphi)\defi\square \varphi$ is equivalent to $\varphi \wedge  \square \varphi \wedge \sneg\square\sneg \varphi \land \sneg \neg \varphi$;\\
(2) $\theta_{t_1}(\varphi)\defi \varphi \land \neg\varphi$ is equivalent to $\varphi   \wedge \sneg \square \varphi \wedge  \sneg\square\sneg \varphi \land \neg \varphi$;\\
(3) $\theta_{t_0}(\varphi)\defi \varphi \land \sneg\neg \varphi \land \sneg\square \varphi$ is equivalent to $\varphi  \wedge \sneg \square \varphi \wedge  \sneg\square\sneg \varphi \land \sneg\neg \varphi$;\\
(4) $\theta_{f_1}(\varphi)\defi \sneg \varphi \land \neg \varphi \land \sneg\square \sneg \varphi$ is equivalent to $\sneg \varphi  \wedge \sneg\square \varphi \wedge \sneg\square\sneg \varphi \land \neg \varphi$;\\
(5) $\theta_{f_0}(\varphi)\defi \sneg \varphi \land \sneg\neg \varphi$ is equivalent to $\sneg \varphi  \wedge \sneg\square \varphi \wedge \sneg\square\sneg \varphi \land \sneg\neg \varphi$;\\
(6) $\theta_{F_1}(\varphi)\defi \square \sneg \varphi$ is equivalent to $\sneg \varphi \wedge \sneg\square \varphi \wedge \square\sneg \varphi \land \neg \varphi$.
\end{proposition}
\begin{proof}
(1) By \axT\ and the properties of classical logic, $\square \varphi$ is equivalent to $\varphi \wedge  \square \varphi$. By \axT\ again and by classical logic, the latter is equivalent to $\varphi \wedge  \square \varphi \land \sneg\square\sneg \varphi$.  By Proposition~\ref{propbaslog}(1), $\square \varphi$ is equivalent to $\square \varphi \land \sneg\neg\varphi$. Putting all together, and by classical logic, $\square \varphi$ is equivalent to  $\varphi \wedge  \square \varphi \wedge \sneg\square\sneg \varphi \land \sneg \neg \varphi$. The proof of the other items is similar, and is left to the reader.
\end{proof}

\noindent Semantically, each formula $\theta_a(p)$  characterizes the truth-value $a$ in the following sense:
 
\begin{proposition} \label{ident-snap}
Let $a \in \{T_0,t_1,t_0, f_1, f_0, F_1\}$. Then, the formulas $\theta_a(p)$ satisfy the following property: for every formula $\varphi$ and for every valuation $v$ over the Nmatrix $\mathcal{M}_{\INT}$, $v(\theta_a(\varphi)) \in \textrm{D}$ iff $v(\varphi)=a$.
\end{proposition}
\begin{proof}
Immediate, from the truth-table below and by the fact that $v(\varphi \land \psi) \in \textrm{D}$ iff  $v(\varphi) \in \textrm{D}$ and $v(\psi) \in \textrm{D}$ for every valuation $v$ over  $\mathcal{M}_{\INT}$.

{\scriptsize
\begin{center}
\begin{tabular}{|c|c|c|c|c|c|c|c|} \hline
$p$  & $\neg p$  & $\sneg p$  & $\sneg\neg p$  & $\square p$ & $\sneg \square p$ & $\square \sneg p$ & $\sneg \square \sneg p$\\[1mm]
 \hline 
    $T_0$   & $F_1$  & $F_1$ & $T_0$ & \textrm{D}  & \textrm{ND} & \textrm{ND} & \textrm{D}\\[1mm] \hline
     $t_0$  & $f_1$ & $f_0, f_1$ & $t_0, t_1$ &   \textrm{ND} & \textrm{D} & \textrm{ND} & \textrm{D} \\[1mm] \hline
     $t_1$  & $t_1$  & $f_0, f_1$ & $f_0, f_1$   & \textrm{ND} & \textrm{D}  & \textrm{ND} & \textrm{D}  \\[1mm] \hline
     $f_0$  & $f_0$ & $t_0, t_1$ &  $t_0, t_1$  & \textrm{ND} & \textrm{D}  & \textrm{ND}  & \textrm{D}\\[1mm] \hline
     $f_1$  & $t_0$ & $t_0, t_1$ & $f_0, f_1$ & \textrm{ND} & \textrm{D}  & \textrm{ND} & \textrm{D}\\[1mm] \hline
     $F_1$  & $T_0$  & $T_0$ & $F_1$  & \textrm{ND} & \textrm{D}  & \textrm{D} & \textrm{ND}\\[1mm] \hline
\end{tabular}
\end{center}}
\end{proof}

\begin{proposition}  \label{Delta-satT}
Let $\Delta$ be a $\delta$-saturated set in \INT. Then, it satisfies the following properties:\\
(1) It is a closed theory in \INT, that is: $\psi \in \Delta$ iff $\Delta \vdash_{\INT} \psi$.\\
(2) $\sneg\psi \in \Delta$ iff $\psi \notin \Delta$.\\
(3) $\neg\neg\psi \in \Delta$ iff $\psi \in \Delta$.\\
(4) $\psi \land \gamma \in \Delta$ iff $\psi \in \Delta$ and  $\gamma \in \Delta$.\\
(5) $\psi \lor \gamma \in \Delta$ iff either $\psi \in \Delta$ or $\gamma \in \Delta$.\\
(6) $\psi \to \gamma \in \Delta$ iff either $\psi \notin \Delta$ or $\gamma \in \Delta$.\\
(7) $\neg(\psi \lor \gamma) \in \Delta$ iff $\neg\psi \in \Delta$ and  $\neg\gamma \in \Delta$.\\
(8) $\neg(\psi \land \gamma) \in \Delta$ iff either $\neg\psi \in \Delta$ or $\neg\gamma \in \Delta$.\\
(9) $\neg(\psi \to \gamma) \in \Delta$ iff $\psi \in \Delta$ and $\neg\gamma \in \Delta$.\\
(10) If	$\square \varphi \in \Delta$ then $\varphi \in \Delta$.\\
(11) $\square \varphi \in \Delta$ iff $\square \sneg \sneg \varphi \in \Delta$.\\
(12) $\square \sneg \neg \varphi \in \Delta$ iff $\square \varphi \in \Delta$.\\
(13) $\square \sneg \neg\neg \varphi \in \Delta$ iff $\square \sneg\varphi \in \Delta$.\\
(14) $\square(\varphi \land \psi) \in \Delta$ iff $\square\varphi \in \Delta$ and $\square\psi \in \Delta$. \\
(15) If $\square\varphi \in \Delta$ or $\square\psi \in \Delta$ then  $\square(\varphi \lor \psi) \in \Delta$. \\
(16) If $\square (\varphi \to \psi) \in \Delta$ then  $\square \varphi \to \square \psi \in \Delta$ and  $\square \sneg\psi \to \square \sneg \varphi \in \Delta$.\\
(17) $\square \sneg(\varphi \vee \psi) \in \Delta$ iff  $\square\sneg\varphi \in \Delta$ and $\square\sneg\psi \in \Delta$. \\
(18) If $\square\sneg\varphi \in \Delta$ or $\square\psi \in \Delta$ then $\square(\varphi \to \psi) \in \Delta$. \\
(19) If $\square\sneg\varphi \in \Delta$ or $\square\sneg\psi \in \Delta$ then $\square \sneg(\varphi \land \psi) \in \Delta$. \\
(20) $\square \sneg(\varphi \to \psi) \in \Delta$ iff $\square \varphi \in \Delta$ and $\square \sneg \psi \in \Delta$.\\
(21) For every $\psi$, exactly one of the formulas $\theta_a(\psi)$ belongs to $\Delta$, for $a \in \{T_0,t_1,t_0, f_1, f_0, F_1\}$.
\end{proposition}
\begin{proof}
Item (1) holds by Remarks~\ref{saturated-rem}(1). Items~(2)-(20) follow by the axioms and rules of the Hilbert calculus for \INT\ (recall Definition~\ref{IvNelLHil}) and by Proposition~\ref{propbaslog}. With respect to item~(21), let $\psi$ be a formula.\\
{\bf Case 1:}  $\square\psi \in \Delta$. This means that, $\theta_{T_0}(\psi) \in \Delta$.\\
{\bf Case 2:} $\square\psi \notin \Delta$.\\
{\bf Case 2.1} $\psi \in \Delta$.\\
{\bf Case 2.1.1:}  $\neg\psi \in \Delta$. By item~(4), $\psi \land \neg\psi \in \Delta$ and so $\theta_{t_1}(\psi) \in \Delta$.\\
{\bf Case 2.1.2}  $\neg\psi \notin \Delta$. By items~(2) and~(4), $\theta_{t_0}(\psi) \in \Delta$.\\
{\bf Case 2.2} $\psi \notin \Delta$.\\
{\bf Case 2.2.1:}  $\square\sneg\psi \in \Delta$. Hence, $\theta_{F_1}(\psi) \in \Delta$.\\
{\bf Case 2.2.2:}  $\square\sneg\psi \notin \Delta$. By item~(2) it follows that $\sneg\square\sneg\psi \in \Delta$ and $\sneg\psi \in \Delta$, since $\psi \notin \Delta$.\\
{\bf Case 2.2.2.1:}  $\neg\psi \in \Delta$. By items~(2) and~(4), $\theta_{f_1}(\psi) \in \Delta$.\\
{\bf Case 2.2.2.2:}  $\neg\psi \notin \Delta$.  By items~(2) and~(4), $\theta_{f_0}(\psi) \in \Delta$.

This shows that, for every $\psi$, $\theta_a(\psi) \in \Delta$ for some $a \in \{T_0,t_1,t_0, f_1, f_0, F_1\}$. By Proposition~\ref{equiv-compo},  if $\theta_a(\psi) \in \Delta$ and  $\theta_b(\psi) \in \Delta$ then $a=b$, given that  $\Delta$ does not contain any contradiction w.r.t. the classical negation $\sneg$.
This completes the proof.
\end{proof}

Before stating the Truth lemma, which is fundamental for the proof of completeness, it will be convenient to stipulate some notation.

\begin{notation} \label{notat}
(1) for $x \in \{0,1\}$ and a formula $\varphi$ the expression $x\varphi$ will denote the formula $\varphi$, if $x=1$, and $\sneg\varphi$, if $x=0$.\\
(2) Let $\gamma_1(p) \defi p$; $\gamma_2(p) \defi \square p$; $\gamma_3(p) \defi \square\sneg p$; and $\gamma_4(p) \defi \neg p$. \\
(3) Let $a=(a_1,a_2,a_3,a_4) \in \textsc{B}_{T}$. Then, $\theta_a(\varphi)$ will be denoted by $\bigwedge_{i=1}^4 a_i \gamma_i(\varphi)$.
\end{notation}

\begin{proposition} [Truth lemma] \label{TruthL} Let $\Delta$ be a $\delta$-saturated set in \INT. Let $v_\Delta:For(\Sigma) \to B$ be a function defined as follows: $v_\Delta(\psi)=a$ iff $\theta_a(\psi) \in \Delta$. Then, $v_\Delta$ is a well-defined function, and it is a valuation over the Nmatrix $\mathcal{M}_{\INT}$ such that, for every formula $\psi$, $v_\Delta(\psi) \in \textrm{D}$ iff $\psi \in \Delta$.
\end{proposition}
\begin{proof}
The function $v_\Delta$ is well-defined, because of  Proposition~\ref{Delta-satT}(21). By definition,  by using Notation~\ref{notat} and by Proposition~\ref{Delta-satT}(4),  $v_\Delta(\varphi)=a$ iff, for $1 \leq i \leq 4$ $a_i \gamma_i(\varphi) \in \Delta$. That is, by Proposition~\ref{Delta-satT}(2),
$$v_\Delta(\varphi)=a \ \mbox{ iff, for $1 \leq i \leq 4$: } \ \gamma_i(\varphi) \in \Delta \ \mbox{ iff } \ a_i=1.$$
Now, it will be proven that  $v_\Delta$ is a valuation over the Nmatrix  $\mathcal{M}_{\INT}$.\\[1mm]
(Negation $\sneg$):  Let $v_\Delta(\varphi)=a$ and $v_\Delta(\sneg \varphi)=b$. Then, by using Notation~\ref{notat}:
\begin{itemize}
\item[-] $\gamma_1(\sneg \varphi)=\sneg \varphi= \sneg \gamma_1(\varphi)$, hence $\gamma_1(\sneg \varphi) \in \Delta$ iff $\gamma_1(\varphi) \notin \Delta$. From this,  $b_1= \sneg a_1$.
\item[-] $\gamma_2(\sneg \varphi)= \square \sneg \varphi= \gamma_3(\varphi)$, hence $\gamma_2(\sneg \varphi) \in \Delta$ iff $\gamma_3(\varphi) \in \Delta$. From this, $b_2= a_3$.
\item[-] $\gamma_3(\sneg \varphi)= \square \sneg\sneg \varphi \equiv\square \varphi= \gamma_2(\varphi)$, hence $\gamma_3(\sneg \varphi) \in \Delta$ iff $\gamma_2(\varphi) \in \Delta$. From this, $b_3= a_2$.
\item[-] $\gamma_4(\sneg \varphi)= \neg \sneg \varphi$, $\gamma_2(\varphi)=\square \varphi$ and $\gamma_3(\varphi)=\square \sneg\varphi$. By Proposition~\ref{propbaslog}(9),  $\square\varphi \to \neg\sneg \varphi \in \Delta$. Hence: 
 if $\gamma_2(\varphi) \in \Delta$ then $\gamma_4(\sneg \varphi) \in \Delta$. That is, $a_2=1$ implies that $b_4=1$ and so $a_2 \leq  b_4$. Now, by Proposition~\ref{propbaslog}(10), $\neg\sneg \varphi \to \sneg\square\sneg \varphi \in \Delta$. Hence, if $\gamma_4(\sneg \varphi) \in \Delta$ then $\sneg \gamma_3(\varphi) \in \Delta$ and so $\gamma_3(\varphi) \notin \Delta$, by  Proposition~\ref{Delta-satT}(2). That is, $b_4=1$ implies that $a_3=0$ and so $\sneg a_3=1$; hence, $b_4 \leq \sneg a_3$.
\end{itemize}
This  shows that $v_\Delta(\sneg \varphi) \in \sneg v_\Delta(\varphi)$.\\[1mm]
(Negation $\neg$): Let $v_\Delta(\varphi)=a$ and $v_\Delta(\neg \varphi)=b$. Then, by reasoning as above:
\begin{itemize}
\item[-] $\gamma_1(\neg \varphi)=\neg \varphi= \gamma_4(\varphi)$, hence $b_1= a_4$.
\item[-] $\gamma_2(\neg \varphi)= \square \neg \varphi \equiv  \square \sneg \varphi= \gamma_3(\varphi)$, hence $b_2= a_3$.
\item[-] $\gamma_3(\neg \varphi)= \square \sneg\neg \varphi \equiv\square \varphi= \gamma_2(\varphi)$, hence $b_3= a_2$.
\item[-] $\gamma_4(\neg \varphi)= \neg \neg \varphi \equiv \varphi= \gamma_1(\varphi)$, hence $b_4= a_1$.
\end{itemize}
This  shows that $v_\Delta(\neg \varphi) \in \neg v_\Delta(\varphi)$.\\[1mm]
(Necessity $\square$): Let $v_\Delta(\varphi)=a$ and $v_\Delta(\square \varphi)=b$. Then:
\begin{itemize}
\item[-] $\gamma_1(\square \varphi)=\square \varphi=\gamma_2(\varphi)$, hence $b_1= a_2$.
\end{itemize}
This  shows that $v_\Delta(\square \varphi) \in \square v_\Delta(\varphi)$.\\[1mm]
(Conjunction $\land$): Let $v_\Delta(\varphi)=a$, $v_\Delta(\psi)=b$ and $v_\Delta(\varphi \land \psi)=c$. Then:
\begin{itemize}
\item[-] $\gamma_1(\varphi \land \psi)=\varphi \land \psi= \gamma_1(\varphi) \land \gamma_1(\psi)$, hence $c_1= a_1 \sqcap b_1$.
\item[-] $\gamma_2(\varphi \land \psi)= \square(\varphi \land \psi) \equiv \square\varphi \land \square \psi= \gamma_2(\varphi) \land \gamma_2(\psi)$, hence $c_2= a_2 \sqcap b_2$.
\item[-] $\gamma_3(\varphi \land \psi)= \square\sneg(\varphi \land \psi)$, $\gamma_3(\varphi)= \square\sneg\varphi$, $\gamma_3(\psi)= \square\sneg\psi$, $\gamma_2(\varphi)= \square\varphi$, $\gamma_2(\psi)= \square\psi$, $\gamma_1(\varphi)= \varphi$, and $\gamma_1(\psi)= \psi$.
Since  $\beta=(\square\sneg\varphi \lor \square\sneg \psi) \to \square\sneg (\varphi \land \psi)$ is provable in \INT, by \axnC, it follows that $\beta \in \Delta$. That is, $(\gamma_3(\varphi) \vee \gamma_3(\psi)) \to \gamma_3(\varphi \vee \psi) \in \Delta$. From this, if either  $\gamma_3(\varphi)  \in \Delta$ or $\gamma_3(\psi)  \in \Delta$ then $\gamma_3(\varphi \vee \psi) \in \Delta$. That is,  $a_3 \sqcup b_3 \leq c_3$. Now, $\square\sneg(\varphi \land \psi) \to \sneg(\varphi \land \psi) \in \Delta$, by \axT, and $\sneg(\varphi \land \psi)  \to (\varphi \to \sneg\psi) \in \Delta$, by \cplp, thus $\square\sneg(\varphi \land \psi) \to (\varphi \to  \sneg\psi) \in \Delta$, by \cplp. Hence, if $\gamma_3(\varphi \land \psi) \in \Delta$ and $\gamma_1(\varphi) \in \Delta$ then $\gamma_1(\psi) \not\in \Delta$. That is, $c_3=1$ and $a_1=1$ implies that $\sneg b_1=1$ or, equivalently, $c_3 \leq a_1 \Rightarrow \sneg b_1$.  In turn, by axioms \axNN\ and \axNT, $\square \sneg(\varphi \land \psi) \to ((\square \varphi \to \square\sneg\psi) \land (\square \psi \to \square\sneg\varphi)) \in \Delta$. By reasoning as above, this implies that $c_3 \leq (a_2 \Rightarrow b_3) \sqcap (b_2 \Rightarrow a_3)$. From this,
$$a_3 \sqcup b_3 \leq c_3 \leq (a_1 \Rightarrow \sneg b_1) \sqcap (a_2 \Rightarrow b_3) \sqcap (b_2 \Rightarrow a_3).$$
\item[-] $\gamma_4(\varphi \land \psi)= \neg (\varphi \land \psi) \equiv \neg \varphi \vee \neg \psi = \gamma_4(\varphi) \vee  \gamma_4(\psi)$, hence $c_4= a_4 \sqcup b_4$.
\end{itemize}
This  shows that $v_\Delta(\varphi \land \psi) \in v_\Delta(\varphi) \land v_\Delta(\psi)$.\\[1mm]
(Disjunction $\lor$): Let $v_\Delta(\varphi)=a$, $v_\Delta(\psi)=b$ and $v_\Delta(\varphi \lor \psi)=c$. Then:
\begin{itemize}
\item[-] $\gamma_1(\varphi \lor \psi)=\varphi \lor \psi= \gamma_1(\varphi) \lor \gamma_1(\psi)$, hence $c_1= a_1 \sqcup b_1$.
\item[-] $\gamma_2(\varphi \lor \psi)= \square(\varphi \lor \psi)$. But  $(\square \varphi \lor \square \psi) \to \square (\varphi \lor  \psi) \in \Delta$, by \axD. That is, $(\gamma_2(\varphi) \vee \gamma_2(\psi)) \to \gamma_2(\varphi \vee \psi) \in \Delta$. Hence, if either  $\gamma_2(\varphi)  \in \Delta$ or $\gamma_2(\psi)  \in \Delta$ then $\gamma_2(\varphi \vee \psi) \in \Delta$. That is,  $a_2 \sqcup b_2 \leq c_2$.  Now,  by \axT, $\square(\varphi \lor \psi) \to (\varphi \lor \psi) \in \Delta$, thus $c_2 \leq a_1 \sqcup b_1$.  In turn, by  axioms \axNo\ and \axNde, $\square (\varphi \lor \psi) \to  ((\square \sneg\varphi \to \square\psi) \land (\square \sneg\psi \to \square\varphi)) \in \Delta$. This implies that $c_2 \leq (a_3 \Rightarrow b_2) \sqcap (b_3 \Rightarrow a_2)$. Hence,
$$a_2 \sqcup b_2 \leq c_2 \leq (a_1 \sqcup b_1) \sqcap (a_3 \Rightarrow b_2) \sqcap (b_3 \Rightarrow a_2).$$
\item[-] $\gamma_3(\varphi \lor \psi)= \square\sneg(\varphi \lor \psi) \equiv  \square\sneg\varphi \land \square\sneg\psi =  \gamma_3(\varphi) \land \gamma_3(\psi)$. That is,  $c_3=a_3 \sqcap b_3$.
\item[-] $\gamma_4(\varphi \lor \psi)= \neg (\varphi \lor \psi) \equiv \neg \varphi \land \neg \psi = \gamma_4(\varphi) \land  \gamma_4(\psi)$, hence $c_4= a_4 \sqcap b_4$.
\end{itemize}
This  shows that $v_\Delta(\varphi \lor \psi) \in v_\Delta(\varphi) \lor v_\Delta(\psi)$.\\[1mm]
(Implication $\to$): Let $v_\Delta(\varphi)=a$, $v_\Delta(\psi)=b$ and $v_\Delta(\varphi \lor \psi)=c$. Then:
\begin{itemize}
\item[-] $\gamma_1(\varphi \to \psi)=\varphi \to \psi= \gamma_1(\varphi) \to \gamma_1(\psi)$, hence $c_1= a_1 \Rightarrow b_1$.
\item[-] $\gamma_2(\varphi \to \psi)= \square(\varphi \to \psi)$. But  $\square(\varphi \to \psi) \to ((\varphi \to \psi) \land (\square \varphi \to \square \psi) \wedge (\square \sneg \psi \to \square \sneg \varphi)) \in \Delta$ and $(\square\sneg\varphi \vee \square\psi) \to \square(\varphi \to \psi) \in \Delta$, since they are provable in \INT. That is, $(\gamma_3(\varphi) \vee \gamma_2(\psi)) \to \gamma_2(\varphi \to \psi) \in \Delta$, and   $\gamma_2(\varphi \to \psi) \to ((\gamma_1(\varphi) \to \gamma_1(\psi)) \land (\gamma_2(\varphi) \to \gamma_2(\psi)) \land (\gamma_3(\psi) \to \gamma_3(\varphi))) \in \Delta$. Hence, $a_3 \sqcup b_2 \leq c_2 \leq (a_1 \Rightarrow b_1) \sqcap (a_2 \Rightarrow b_2) \sqcap(b_3 \Rightarrow a_3)$.
\item[-] $\gamma_3(\varphi \to \psi)= \square\sneg(\varphi \to \psi) \equiv  \square\varphi \land \square\sneg\psi =  \gamma_2(\varphi) \land \gamma_3(\psi)$. That is,  $c_3=a_2 \sqcap b_3$.
\item[-] $\gamma_4(\varphi \to \psi)= \neg (\varphi \to \psi) \equiv \varphi \land \neg \psi = \gamma_1(\varphi) \land  \gamma_4(\psi)$, hence $c_4= a_1 \sqcap b_4$.
\end{itemize}
This  shows that $v_\Delta(\varphi \to \psi) \in v_\Delta(\varphi) \to v_\Delta(\psi)$.\\

From the facts above, it follows that the function $v_\Delta$ is a valuation over the Nmatrix $\mathcal{M}_{\INT}$. Moreover, by the very definitions, $v_\Delta(\psi) \in \textrm{D}$ iff $\psi \in \Delta$, for every formula $\psi$.
\end{proof}

\begin{theorem} [Completeness]
Let $\Gamma \cup \{\varphi\}$ be a set of formulas. Then, $\Gamma \models_{\INT} \varphi$ implies that $\Gamma \vdash_{\INT} \varphi$.
\end{theorem}
\begin{proof}
Suppose that $\Gamma \nvdash_{\INT} \varphi$. Then, by Remarks~\ref{saturated-rem}(2), there exists a $\varphi$-saturated set $\Delta$ such that $\Gamma \subseteq \Delta$. Let $v_\Delta$ be the valuation over the Nmatrix $\mathcal{M}_{\INT}$ defined as in Proposition~\ref{TruthL}. Then, $v_\Delta(\psi) \in \textrm{D}$ for every $\psi \in \Gamma$, but $v_\Delta(\varphi) \notin \textrm{D}$. This shows that $\Gamma \not\models_{\INT} \varphi$.
\end{proof}

\section{\INT\ as a combined logic: conservativeness results} \label{conserva}

At the begining of Section~\ref{sectINT} it was presented a technique for combining, under certain coherence hypothesis, two finite Nmatrices presented as swap structures over $\A_2$. The logic \INT\ was obtained by applying this technique, and so this logic  can be seen as the result of a process of {\em combination  of logics}, specifically a combination of logic \pyn\ with logic \Tm. 

Combining logics is a relatively new field in contemporary logic.\footnote{General references on combining logics can be found, for instance, in~\cite{car:con:07b} and~\cite{car:con:gab:gou:ser:08}.} In a  bottom-up perspective, a process for combining logics produces a new one, which is intended to be minimal in some sense. Hence, if a logic {\bf L} is a combination of given logics ${\bf L}_1$ and ${\bf L}_2$  it should be expected that: (1) {\bf L} extends both ${\bf L}_1$ and ${\bf L}_2$; and (2) {\bf L} is a minimal extension of both ${\bf L}_1$ and ${\bf L}_2$ (in some sense). For instance, some methods (as, for instance, the well-known method of {\em fibring}) may require {\bf L} to be the least conservative extension of both ${\bf L}_1$ and ${\bf L}_2$ (see~\cite{gab:99}).

In order to support the idea that \INT\ is obtained from \pyn\ and \Tm\ by a process of combination of logics, it will be proven in this section that \INT\ is a conservative expansion of both logics (see Theorems~\ref{Int-conserv-Pyn} and~\ref{conservaTm} below).

\begin{definition}
Let $\mathcal{M}_1=\langle A_1,D_1,\mathcal{O}_1\rangle$ and $\mathcal{M}_2=\langle A_2,D_2,\mathcal{O}_2\rangle$ be Nmatrices over a signature $\Theta$. We say that $\mathcal{M}_1$ is a subNmatrix of $\mathcal{M}_2$ if $A_1 \subseteq A_2$, $D_1=D_2 \cap A_1$ and, for every connective $\# $ of $\Theta$ and every $z \in A_1^n$, $\mathcal{O}_1(\#)(z) \subseteq \mathcal{O}_2(\#)(z)$, where $n$ is the arity of $\#$. In particular, $\mathcal{O}_1(\#) \subseteq \mathcal{O}_2(\#)$ for every $0$-ary connective $\#$.
\end{definition}

\begin{definition}\label{expaB4}
Let $\mathcal{A}_4^\to=\langle A_4,\mathcal{O}_4^\to\rangle$ be the algebra over $\Sigma_4$, considered as a multialgebra, underlying  the  characteristic matrix $\mathcal{M}_4=\langle A_4,D_4,\mathcal{O}_4^\to\rangle$ of \pyn, where $A_4=\{{\bf 1},\bo,\nei,{\bf 0}\}$ and $D_4=\{{\bf 1},\bo\}$ (recalling that ${\bf 1}=(1,0)$; $\bo=(1,1)$; $\nei=(0,0)$;  and ${\bf 0}=(0,1)$).
Let $h:A_4 \to \textsc{B}_{T}$ be a function such that $h(z_1,z_2)=(z_1,0,0,z_2)$ (that is, $h({\bf 1})=t_0$; $h(\bo)=t_1$; $h(\nei)=f_0$; and $h({\bf 0})=f_1$). 
Let $h(\mathcal{A}_4^{\to})\defi\langle h(A_4),\mathcal{O}'_4\rangle$ be the multialgebra over $\Sigma_4$ induced by $h$ and $\mathcal{A}_4^{\to}$, where $h(A_4)=\{t_0,t_1,f_0,f_1\}$. That is: $\mathcal{O}'_4(\neg)h(z)\defi h(\mathcal{O}_4^{\to}(\neg)(z))$ for every  $z \in A_4$, and $h(z)\mathcal{O}'_4(\#)h(w)\defi h(z\mathcal{O}_4^{\to}(\#)w)$ for every $\# \in \{\land,\lor,\to\}$ and $z,w \in A_4$.\footnote{This is well-defined, since $h$ is a bijection from $A_4$ to $h(A_4)$.} Let $h(\mathcal{M}_4)\defi \langle h(A_4),h(D_4),\mathcal{O}'_4\rangle$ be the  Nmatrix over $\Sigma_4$  induced by  $h$ and $\mathcal{M}_4^{\to}$, where $h(D_4)=\{t_0,t_1\}$.
\end{definition}

\begin{proposition} \label{expB4sub}
$h(\mathcal{M}_4)$ is a subNmatrix of the reduct of $\mathcal{M}_{\INT}$ to $\Sigma_4$.
\end{proposition}
\begin{proof}
Straightforward, from the definitions.
\end{proof}

\begin{lemma} \label{lemma-M4}
Let $v:For(\Sigma_4) \to A_4$ be a valuation over the Nmatrix $\mathcal{M}_4$ of \pyn. Then, there exists a valuation  $\bar{v}:For(\Sigma) \to \textsc{B}_{T}$ over the Nmatrix  $\mathcal{M}_{\INT}$ such that $\bar{v}(\varphi)=h(v(\varphi))$ for every $\varphi \in For(\Sigma_4)$. Moreover, $\bar{v}(\delta) \in \textrm{D}$ iff $v(\delta) \in D_4$, for every $\delta \in For(\Sigma_4)$. 
\end{lemma}
\begin{proof} 
Recall that $\mathcal{M}_{\INT}= \langle \textsc{B}_{T}, \textrm{D}, \mathcal{O}\rangle$ such that $\mathcal{O}(\#)$ is denoted by $\tilde{\#}$ for $\# \in \Sigma$.
Define $\bar{v}(\varphi)=h(v(\varphi))$ for every $\varphi \in For(\Sigma_4)$;  $\bar{v}(\#\varphi) \in \mathcal{O}(\#)(\bar{v}(\varphi))$  for every  $\# \in \{\neg, \sneg,\square\}$ and $\varphi \notin For(\Sigma_4)$; and $\bar{v}(\varphi \# \psi) =\bar{v}(\varphi) \mathcal{O}(\#) \bar{v}(\psi)$ for $\# \in \{\land, \vee,\to\}$, if $\{\varphi,\psi\} \not\subseteq For(\Sigma_4)$. Then, $\bar{v}$ satisfies the requirements. Indeed, for every $\delta,\psi \in For(\Sigma)$, by Definition~\ref{expaB4} and by  Proposition~\ref{expB4sub}, we have the following:

\begin{itemize}
\item[-] Let $\#\delta \in For(\Sigma)$, where $\# \in \{\neg,\sneg,\square\}$. If  $\#\delta \in For(\Sigma_4)$ then $\#=\neg$. Since $v(\neg\delta) \in  \mathcal{O}_4^{\to}(\neg)(v(\delta))$ then 

$\begin{array}{lll}
\bar{v}(\neg\delta)&=&h(v(\neg\delta)) \in  h(\mathcal{O}_4^{\to}(\neg)(v(\delta)))=\mathcal{O}'_4(\neg)(h(v(\delta)))\\
&=& \mathcal{O}'_4(\neg)(\bar{v}(\delta)) \subseteq \tilde{\neg}(\bar{v}(\delta)).
\end{array}$\\
Now, if $\#\delta \notin For(\Sigma_4)$ then, by definition, $\bar{v}(\#\delta) \in \tilde{\#}(\bar{v}(\delta))$.

\item[-] Let $\delta \# \psi \in For(\Sigma)$, where $\# \in \{\land,\vee,\to\}$. If $\delta \# \psi \in For(\Sigma_4)$ then 
$v(\delta \# \psi) \in  v(\delta)\mathcal{O}_4^{\to}(\#)v(\psi)$, hence

$\begin{array}{lll}
\bar{v}(\delta \# \psi)&=&h(v(\delta \# \psi)) \in  h(v(\delta)\mathcal{O}_4^{\to}(\#)v(\psi))\\
&=& h(v(\delta))\mathcal{O}'_4(\#)h(v(\psi)) = \bar{v}(\delta)\mathcal{O}'_4(\#)\bar{v}(\psi) \subseteq \bar{v}(\delta) \tilde{\#} \bar{v}(\psi).
\end{array}$\\
Now, if $\delta \# \psi \notin For(\Sigma_4)$ then, by definition, $\bar{v}(\delta \# \psi) \in \bar{v}(\delta) \tilde{\#} \bar{v}(\psi)$.
\end{itemize}
This shows that $\bar{v}$ is a valuation over the Nmatrix  $\mathcal{M}_{\INT}$ such that, for every $\delta \in For(\Sigma_4)$, $\bar{v}(\delta)=h(v(\delta)) \in \textrm{D}$ iff $v(\delta) \in D_4$. 
\end{proof}

\begin{theorem} \label{Int-conserv-Pyn}
The logic \INT\ is a conservative expansion of \pyn. That is, for every set of formulas $\Gamma \cup \{\varphi\}$ in the signature of  \pyn, $\Gamma \vdash_{\INT} \varphi$ \ iff \  $\Gamma \vdash_{\pyn} \varphi$.
\end{theorem}
\begin{proof} 
By Definition, the Hilbert calculus for \INT\ contains all the schema axioms and inference rules of \pyn. Hence, for every set of formulas $\Gamma \cup \{\varphi\}$ in the signature of  \pyn, $\Gamma \vdash_{\pyn} \varphi$ implies that  $\Gamma \vdash_{\INT} \varphi$. 

The converse will be proven semantically, by showing that $\Gamma \models_{\INT} \varphi$  implies that $\Gamma \models_{\pyn} \varphi$, for every set of formulas $\Gamma \cup \{\varphi\}$ in the signature $\Sigma_4$ of  \pyn. Thus, let  $\Gamma \cup \{\varphi\} \subseteq For(\Sigma_4)$ be such that  $\Gamma \models_{\INT} \varphi$, and let $v:For(\Sigma_4) \to A_4$ be a valuation over the Nmatrix $\mathcal{M}_4$ of \pyn\ such that $v(\gamma) \in D_4$ for every $\gamma \in \Gamma$. Let  $\bar{v}:For(\Sigma) \to \textsc{B}_{T}$ be the valuation over the Nmatrix  $\mathcal{M}_{\INT}$  defined from $v$ as in Lemma~\ref{lemma-M4}. Then, for every $\delta \in For(\Sigma_4)$, $\bar{v}(\delta) \in \textrm{D}$ iff  $v(\delta) \in D_4$. From this, $\bar{v}(\gamma)  \in \textrm{D}$  for every $\gamma \in \Gamma$. Since it is assumed that $\Gamma \models_{\INT} \varphi$, it follows that  $\bar{v}(\varphi)  \in \textrm{D}$ and so $v(\varphi) \in D_4$. This shows that $\Gamma \models_{\pyn} \varphi$, as required.
\end{proof}

\noindent Now, an analogous result will be proven for \Tm.

\begin{definition} \label{expaTm}
Let $\mathcal{A}_T=\langle A_m,\mathcal{O}_T\rangle$ be the multialgebra over $\Sigma_m$ underlying  the  characteristic Nmatrix $\mathcal{M}_T=\langle {A}_m,D_m,\mathcal{O}_T\rangle$ of \Tm\ (see Definitions~\ref{Tm-mat-def} and~\ref{NMatIvlev}), where $A_m=\{T,t,f,F\}$ and $D_m=\{T,t\}$ (recalling that $T=(1,1,0)$; $t=(1,0,0)$; $f=(0,0,0)$;  and $F=(0,0,1)$). Let $g:A_m \to \textsc{B}_{T}$ be such that $g(z_1,z_2,z_3)=(z_1,z_2,z_3,\sneg z_1)$ (that is, $g(T)=T_0$; $g(t)=t_0$; $g(f)=f_1$; and $g(F)=F_1$).
Let $g(\mathcal{A}_T)\defi\langle g(A_m),\mathcal{O}'_T\rangle$ be the multialgebra over $\Sigma_m$ induced by $g$ and $\mathcal{A}_T$, where $g(A_m)=\{T_0,t_0,f_1,F_1\}$. That is: $\mathcal{O}'_T(\#)g(z)\defi g(\mathcal{O}_T(\#)(z))$ for every $\# \in \{\sneg,\square\}$ and $z \in A_m$; and $g(z)\mathcal{O}'_T(\to)g(w)\defi g(z\mathcal{O}_T(\to)w)$ for every $z,w \in A_m$.\footnote{This is well-defined, since $g$ is a bijection from $A_m$ to $g(A_m)$.} Let $g(\mathcal{M}_T)\defi\langle g(A_m),g(D_m),\mathcal{O}'_T\rangle$ be the  Nmatrix over $\Sigma_m$  induced by  $g$ and $\mathcal{M}_T$, where $g(D_m)=\{T_0,t_0\}$.
\end{definition}

\begin{proposition} \label{expTmsub}
$g(\mathcal{M}_T)$ is a subNmatrix of the reduct of $\mathcal{M}_{\INT}$ to $\Sigma_m$.
\end{proposition}
\begin{proof}
Straightforward, from the definitions.
\end{proof}

\begin{lemma} \label{lemma-Tm}
Let $v:For(\Sigma_m) \to A_m$ be a valuation over the Nmatrix $\mathcal{M}_T$ of \Tm. Then, there exists a valuation  $\hat{v}:For(\Sigma) \to \textsc{B}_{T}$ over the Nmatrix   $\mathcal{M}_{\INT}$ such that is  $\hat{v}(\varphi)=g(v(\varphi))$ for every $\varphi \in For(\Sigma_m)$.  Moreover, $\hat{v}(\delta) \in \textrm{D}$ iff $v(\delta) \in D_m$, for every $\delta \in For(\Sigma_m)$. 
\end{lemma}
\begin{proof} 
Define $\hat{v}(\varphi)=g(v(\varphi))$ for every $\varphi \in For(\Sigma_m)$; and $\hat{v}(\#\varphi) \in \mathcal{O}(\#)(\hat{v}(\varphi))$  for every  $\# \in \{\neg, \sneg,\square\}$ and $\varphi \notin For(\Sigma_m)$; and $\hat{v}(\varphi \# \psi) =\hat{v}(\varphi) \mathcal{O}(\#) \hat{v}(\psi)$ for $\# \in \{\land, \vee,\to\}$, if $\{\varphi,\psi\} \not\subseteq For(\Sigma_m)$. Then, with a proof similar to the one given for Lemma~\ref{lemma-M4} (but this time based on  Definition~\ref{expaTm} and Proposition~\ref{expTmsub}) it can be seen that $\hat{v}$ satisfies the desired requirements. Details are left to the reader.
\end{proof}

\begin{theorem} \label{conservaTm}
The logic \INT\ is a conservative expansion of  \Tm. That is, for every set of formulas $\Gamma \cup \{\varphi\}$ in the signature of  \Tm, $\Gamma \vdash_{\INT} \varphi$ \ iff \  $\Gamma \vdash_{\Tm} \varphi$.
\end{theorem}
\begin{proof}
By Definition, the Hilbert calculus for \INT\ contains all the schema axioms and inference rules of \Tm. The only differences are the following:

\begin{itemize}
\item[-]  in axiom \axKd, $\alpha \land \beta$ is an abbreviation of $\sneg(\alpha \to \sneg\beta)$, while in the Hilbert calculus for \INT\ the conjunction $\land$ in \axKd\ is primitive and obeys the axioms of \cpl;
\item[-]  in axiom \axMu, $\alpha \vee \beta$ is an abbreviation of $\sneg\alpha \to \beta$, while in the Hilbert calculus for \INT\ the disjunction $\vee$ in \axMu\ is primitive and obeys the axioms of \cpl;
\item[-]  in axioms \axKd\ and \axDNu,  $\alpha \sse \beta$ is an abbreviation of $\sneg((\alpha \to \beta) \to \sneg(\beta \to \alpha))$, while in  the Hilbert calculus for \INT\ the biconditional $\sse$ in \axKd\ is the abbreviation $\alpha \sse \beta \defi  (\alpha \to \beta) \land (\beta \to \alpha)$, where $\land$ is the primitive conjunction.
\end{itemize}
However, since  the Hilbert calculus for \INT\ contains \cpl\ over the signature $\{\land,\vee, \to, \sneg\}$ it is immediate that axioms \axKd,  \axMu\ and \axDNu\, formulated in the signature $\Sigma_m$, can be derived in the Hilbert calculus for \INT\ from the corresponding version over the signature $\Sigma$. Hence, for every set of formulas $\Gamma \cup \{\varphi\}$ in the signature $\Sigma_m$ of  \Tm, $\Gamma \vdash_{\Tm} \varphi$ implies that  $\Gamma \vdash_{\INT} \varphi$. 

The converse will be proven semantically, by showing that $\Gamma \models_{\INT} \varphi$  implies that $\Gamma \models_{\Tm} \varphi$, for every set of formulas $\Gamma \cup \{\varphi\} \subseteq For(\Sigma_m)$. Thus, let  $\Gamma \cup \{\varphi\} \subseteq For(\Sigma_m)$ be such that  $\Gamma \models_{\INT} \varphi$, and let $v:For(\Sigma_m) \to A_m$ be a valuation over the Nmatrix $\mathcal{M}_T$ of \Tm\ such that $v(\gamma) \in D_m$ for every $\gamma \in \Gamma$. Let  $\hat{v}$  be the valuation over the Nmatrix  $\mathcal{M}_{\INT}$ defined as in Lemma~\ref{lemma-Tm}. Then, for every $\delta \in For(\Sigma_m)$, $\hat{v}(\delta) \in \textrm{D}$ iff $v(\delta) \in D_m$. From this it follows, as in the proof of Theorem~\ref{Int-conserv-Pyn}, that $v(\varphi) \in D_m$. This shows that $\Gamma \models_{\Tm} \varphi$.
\end{proof}

\

\noindent  Theorems~\ref{Int-conserv-Pyn} and~\ref{conservaTm} are of central importance: they reveal that the method of combination of swap structures  by  superposition of snapshots, applied here to obtain \INT, can be considered as a genuine method for combining logics. As we shall see in the next section, the same method can be applied {\em mutatis mutandis} to combine \pyn\ with axiomatic extensions of \Tm. In the examples analyzed in this paper, the method combines a logic characterized by swap structures (namely, \Tm\ or an axiomatic extension of it) with another logic characterized by twist structures (namely, \pyn). However, twist structures are special (deterministic) cases of swap structures. 

The combination process of \Tm\ and \pyn\ shares the classical implication $\to$ of both logics (conjunction and disjunction are not directly shared since they are derived connectives in \Tm). In this sense, our process could be seen as a {\em constrained fibring} of \Tm\ with \pyn\ by sharing implication (see, for instance, \cite[Chapters~2 and~3]{car:con:gab:gou:ser:08}). However, when combining the Hilbert calculi, some new axioms in the combined language (specifically, axioms \axNDN\ and \axNTN) were included. Both axioms, which are sound w.r.t. the 6-valued Nmatrices, allow to derive some basic principles assumed for the combined logic, namely $\square \varphi \to \sneg\neg \varphi$ and $\square\sneg\varphi \to \neg\varphi$, among others (recall Proposition~\ref{propbaslog}).  This kind of axioms in the mixed language, manually added after a basic combination process, are called {\em bridge principles} or {\em bridge axioms}.
The adoption of bridge principles is a fundamental tool when combining logics: sometimes, the axiomatization of the combined logic requires a fine-tuning to meet  the semantical specifications. 

It is worth observing that, when combining logics, for instance  by fibring, the preservation of metaproperties such as soundness and completeness is not an easy task to be established, and usually some additional requirements need to be considered in the logics being combined (for instance, the presence of an implication connective satisfying the deduction metatheorem, see~\cite[Section~3.3]{car:con:gab:gou:ser:08}) in order to guarantee that preservation.
To this respect, it is interesting to observe that, by its own nature, the process of  combination of swap structures proposed here allows to prove the soundness and completeness in a very natural way. Indeed, as happens with swap  structures (and in the particular, deterministic case of twist structures) there is a natural correspondence between the formal specifications of each  multioperator and the axioms to be considered in the Hilbert calculus of the logic being represented, taking into account that each coordinate of the snapshots represents an specific formula depending on one variable $p$ (in the present case: $p$, $\square p$, $\square\sneg p$ and $\neg p$).\footnote{Several instructive examples of this phenomenon, in the context of twist structures, can be found in~\cite{bor:con:her:20}.} Superposition of snapshots reproduces, in the snapshot obtained by superposition, the original formal specifications. This feature is naturally represented axiomatically  just by gathering together the axioms of the logics being combined, plus some specific (and self-evident) bridge principles. 

Finally, concerning the bridge principles considered for \INT, they were motivated by the conceptual nature of the logics being combined, and so they are not compulsory. In other words, it would be perfectly feasible to combine \pyn\ with \Tm\ (as well as its axiomatic extensions to be considered in the next section) without restricting the snapshots, obtaining so an 8-valued Nmatrix. If, in addition conjunction and disjunction were considered as primitive in the signature of \Tm\ (hence, axioms 
\axC-\axD\ and \axNN-\axNde\ would be included in the Hilbert calculus for \Tm\ in the signature $\{\land, \vee, \to, \sneg,\square\}$) then the combination of \Tm\ and \pyn\ without restrictions on the snapshots would be axiomatized by the constrained fibring of the Hilbert calculi of both logics by sharing the signature $\{\land,\vee, \to\}$, without adding any bridge axiom (see~\cite[Section~2.2]{car:con:gab:gou:ser:08}). That is, the axiomatization of the combined logic is obtained just by joining both calculi and by identifying $\{\land,\vee, \to\}$. The resulting logic is still a conservative expansion of both logics, and so the combination of swap structures by superposition of snapshots (without restrictions)  corresponds, in the considered examples, to the constrained fibring of these logics by sharing $\{\land,\vee, \to\}$. Soundness and completeness of the obtained 8-valued Nmatrix w.r.t. the proposed axiomatization can be proven as above. This shows that the method guarantees the preservation of soundness and completeness.

\section{Combining \pyn\ with other Ivlev-like modal logics} \label{other-logics}

As mentioned in the introduction,  in~\cite{con:cerro:per:15} were considered, besides \Tm, other 4-valued Ivlev-like modal logics corresponding to Ivlev-like versions of  modal systems {\bf KT4}, {\bf KT45} and  {\bf KTB}, called respectively 
\Tqm, \Tqsm, and \TBm.\footnote{\Tm\ and \Tqsm\ were considered by Ivlev under the names $Sa^+$ and $Sb^+$, respectively.} Each of this logics is semantically characterized by a single 4-valued Nmatrix obtained from the one of  \Tm\ by modifying the definition of the multioperator interpreting $\square$. In turn, they are axiomatized in a natural way from an axiomatization of \Tm. Specifically, the logic \Tqm\ is defined by adding to the Hilbert calculus of \Tm\ axiom schema\\[2mm]
$\begin{array}{ll}
    \axSq & \square \varphi \to \square \square \varphi\\[2mm]
\end{array}$

\noindent
The logic \Tqsm\ is defined by adding to \Tqm\ axiom schema\\[2mm]
$\begin{array}{ll}
   \axSc & \sneg \square \sneg \square \varphi \to\square \varphi\\[2mm]
\end{array}$

\noindent
The logic \TBm\ is defined by adding to \Tm\ axiom schema\\[2mm]
$\begin{array}{ll}
    \axB & \sneg \square \sneg \square \varphi \to \varphi\\[2mm]
\end{array}$

\noindent
It can be also considered  the logic \TqBm, obtained from \TBm\ by adding axiom schema \axSq, as well as the logic \Tsm, obtained by adding to \Tm\ axiom schema \axSc. In~\cite{con:gol:19} were introduced  swap structures over Boolean algebras characterizing logics \Tm\ (recall Definition~\ref{NMatIvlev}), \Tqm, \Tqsm, and \TBm. The only change  to be made in the multialgebra  $\mathcal{B}_{\A}$ for \Tm\  concerns the interpretation of $\square$, to be respectively defined as follows (here, we are also introducing the interpretation of $\square$ in \TqBm\ and \Tsm):

\begin{itemize}
\item[-] In \Tqm: $\square_4 z = \{u\in  \mathbb{B}_{\mathcal{A}} \ : \ u_1=u_2=z_2\}$;
\item[-] In \Tqsm: $\square_{45} z = \{u\in  \mathbb{B}_{\mathcal{A}} \ : \ u_1=u_2=z_2 \ \mbox{ and } \ u_3 = \sneg z_2\}$;
\item[-] In \TBm: $\square_B z = \{u\in  \mathbb{B}_{\mathcal{A}} \ : \ u_1=z_2 \ \mbox{ and } \ u_3 \geq \sneg z_1\}$;
\item[-] In \TqBm: $\square_{4B} z = \{u\in  \mathbb{B}_{\mathcal{A}} \ : \ u_1=u_2=z_2 \ \mbox{ and } \ u_3 \geq \sneg z_1\}$;
\item[-] In \Tsm: $\square_5 z = \{u\in  \mathbb{B}_{\mathcal{A}} \ : \ u_1=z_2 \ \mbox{ and } \ u_3 = \sneg z_2\}$.
\end{itemize}
In particular, for ${\bf L} \in \{\Tqm, \Tqsm, \TBm, TqBm, \Tsm\}$, the  4-valued Nmatrix $\mathcal{M}_{\bf L}$ for {\bf L} with domain $A_m$ (i.e., defined over the 2-element Boolean algebra $\A_2$) is obtained from $\mathcal{M}_{T}$ by replacing respectively the interpretation of $\square$ by the following:
{\scriptsize
\begin{center}
\begin{tabular}{|c|c|c|c|c|c|} \hline
$\quad$   & $\square_4$  & $\square_{45}$  & $\square_B$  & $\square_{4B}$  & $\square_5$\\[1mm]
 \hline
    $T$     & $T$  & $T$ & $\{T, t\}$ & $T$ & $\{T, t\}$ \\[1mm] \hline
     $t$      & $\{F, f\}$ & $F$ & $\{F, f\}$ & $\{F, f\}$ & $F$\\[1mm] \hline
     $f$     & $\{F, f\}$ & $F$ & $F$ & $F$ & $F$ \\[1mm] \hline
     $F$    & $\{F, f\}$  & $F$ & $F$ & $F$ & $F$\\[1mm] \hline
\end{tabular}
\end{center}}

\

\noindent
By adapting the  swap structures above, and by the same methodology used for defining \INT, it can be obtained the combined logic \INTq, \INTqs, \INTB, \INTqB\ and \INTs\ based, respectively, on \Tqm, \Tqsm, \TBm, \TqBm, and \Tsm. If {\bf L} is one of such combined logics then the Nmatrix  $\mathcal{M}_{\bf L}= \langle \textsc{B}_{T}, \textrm{D}, \mathcal{O}_{\bf L}\rangle$  for {\bf L} over $\Sigma$ is defined as $\mathcal{M}_{\INT}$, where the only difference is with respect to the interpretation of $\square$, which is given as follows:
\begin{itemize}
\item[-] In \INTq: $\tilde{\square}_4 z = \{u\in \textsc{B}_{T} \ : \ u_1=u_2=z_2\}$;
\item[-] In \INTqs: $\tilde{\square}_{45} z = \{u\in \textsc{B}_{T} \ : \ u_1=u_2=z_2 \ \mbox{ and } \ u_3 = \sneg z_2\}$;
\item[-] In \INTB: $\tilde{\square}_B z = \{u\in \textsc{B}_{T} \ : \ u_1=z_2 \ \mbox{ and } \ u_3 \geq \sneg z_1\}$;
\item[-] In \INTqB: $\tilde{\square}_{4B} z = \{u\in \textsc{B}_{T} \ : \ u_1=u_2=z_2 \ \mbox{ and } \ u_3 \geq \sneg z_1\}$;
\item[-] In \INTs: $\tilde{\square}_5 z = \{u\in \textsc{B}_{T} \ : \ u_1=z_2 \ \mbox{ and } \ u_3 = \sneg z_2\}$.
\end{itemize}

\noindent In a more concise notation:

\noindent
(7) $\tilde{\square}_4 z = (z_2,z_2,\_ ,\_ )$ in \INTq;\\[1mm]
(8) $\tilde{\square}_{45} z = (z_2,z_2,\sneg z_2 ,\_ )$ in \INTqs;\\[1mm]
(9) $\tilde{\square}_B z = (z_2,\_,\sneg z_1 \leq \_,\_ )$ in \INTB;\\[1mm]
(10) $\tilde{\square}_{4B} z = (z_2,z_2,\sneg z_1 \leq \_,\_ )$ in \INTqB;\\[1mm]
(11) $\tilde{\square}_5 z = (z_2,\_,\sneg z_2 ,\_ )$ in \INTs.

\noindent This produces the following (non-deterministic) truth-tables for the necessitation operator in the logic \INTq, \INTqs, \INTB, \INTqB\ and \INTs, respectively:

{\scriptsize
\begin{center}
\begin{tabular}{|c|c|c|c|c|c|} \hline
$\quad$   & $\tilde{\square}_4$  & $\tilde{\square}_{45}$  & $\tilde{\square}_B$  & $\tilde{\square}_{4B}$  & $\tilde{\square}_5$\\[1mm]
 \hline
    $T_0$     & $T_0$  & $T_0$ & \textrm{D} & $T_0$ & \textrm{D}\\[1mm] \hline
     $t_0$      & \textrm{ND} & $F_1$ &\textrm{ND} & \textrm{ND}  & $F_1$\\[1mm] \hline
     $t_1$     & \textrm{ND}  & $F_1$ & \textrm{ND}  & \textrm{ND} & $F_1$   \\[1mm] \hline
     $f_0$     & \textrm{ND} & $F_1$ & $F_1$ & $F_1$ & $F_1$ \\[1mm] \hline
     $f_1$    & \textrm{ND} & $F_1$ & $F_1$ & $F_1$ & $F_1$ \\[1mm] \hline
     $F_1$    & \textrm{ND}  & $F_1$ & $F_1$ & $F_1$ & $F_1$\\[1mm] \hline
\end{tabular}
\end{center}}

\noindent From now on, let $\mathbb{L}\defi  \{\INTq, \INTqs, \INTB, \INTqB, \INTs\}$.

\begin{theorem} [Soundness]
Let ${\bf L} \in \mathbb{L}$ and let $\Gamma \cup \{\varphi\}$ be a set of formulas. Then, $\Gamma \vdash_{\bf L} \varphi$ implies that $\Gamma \models_{\bf L} \varphi$.
\end{theorem}
\begin{proof} Given that  $\mathcal{M}_{\bf L}$ is a subNmatrix of $\mathcal{M}_{\INT}$, all the axioms of \INT\ are valid in $\mathcal{M}_{\bf L}$, by Theorem~\ref{sound-INT}. Hence, it is enough to prove that each instance of every modal axiom added to \Tm\ to define {\bf L} is valid in $\mathcal{M}_{\bf L}$. In what follows, let $v$ be a valuation over $\mathcal{M}_{\bf L}$.\\
\INTq: Let $\delta=\square \varphi \to \square \square \varphi$ be an instance of axiom \axSq. Then, $v(\square \varphi)_1=v(\varphi)_2=v(\square \varphi)_2=v(\square \square \varphi)_1$. Hence, $v(\delta)_1=1$.\\
\INTs: Let $\delta=\sneg \square \sneg \square \varphi \to\square \varphi$ be an instance of axiom \axSc. Then, $v(\sneg \square \sneg \square \varphi)_1=\sneg v(\square \sneg \square \varphi)_1= \sneg v(\square \varphi)_3=\sneg \sneg v(\varphi)_2= v(\varphi)_2 = v(\square\varphi)_1$. This shows that $v(\delta)_1=1$.\\
\INTB: Let $\delta=\sneg \square \sneg \square \varphi \to \varphi$ be an instance of axiom \axB. Then, $v(\sneg \square \sneg \square \varphi)_1=\sneg v(\square \sneg \square \varphi)_1=\sneg v(\square \varphi)_3 \leq v(\varphi)_1$, since $v(\square \varphi)_3 \geq \sneg v(\varphi)_1$. From this, $v(\delta)_1=1$.\\
\INTqs\ and \INTqB: The proof is obtained  by combining the previous cases. 

This completes the proof.
\end{proof}

\noindent Now, let us concentrate on the proof of completeness.
For $a \in  \textsc{B}_{T}$ let $\theta_a(p)$ be the formula defined as in Proposition~\ref{equiv-compo}. Then, the following can be proved in an analogous way as it was done for \INT:

\begin{proposition} \label{ident-snap-GEN}
Let ${\bf L} \in \mathbb{L}$. Then the formulas $\theta_a(p)$, for $a \in  \textsc{B}_{T}$, satisfy the following property: for every formula $\varphi$ and for every valuation $v$ over the Nmatrix $\mathcal{M}_{\bf L}$, $v(\theta_a(\varphi)) \in \textrm{D}$ iff $v(\varphi)=a$.
\end{proposition}

\begin{proposition}  \label{Delta-satT-GEN}
Let ${\bf L} \in \mathbb{L}$, and let $\Delta$ be a $\delta$-saturated set in {\bf L}. Then, $\Delta$  is a closed theory in {\bf L} and it satisfies  properties (2)-(21) of Proposition~\ref{Delta-satT}, plus the following:\\
(22) $\square\square\psi \in\Delta$ iff $\square \psi \in \Delta$ (if ${\bf L}= \INTq$).\\
(23) $\square\sneg\square\psi \in\Delta$ iff $\sneg \square \psi \in \Delta$ (if ${\bf L}= \INTs$).\\
(24) If $\psi \notin\Delta$ then $\square\sneg \square \psi \in \Delta$ (if ${\bf L}= \INTB$).\\
(25) $\square\square\psi \in\Delta$ iff $\square \psi \in \Delta$; and $\square\sneg\square\psi \in\Delta$ iff $\sneg \square \psi \in \Delta$  (if ${\bf L}= \INTqs$).\\
(26) If $\psi \notin\Delta$ then $\square\sneg \square \psi \in \Delta$; and $\square\square\psi \in\Delta$ iff $\square \psi \in \Delta$ (if ${\bf L}= \INTqB$).
\end{proposition}
\begin{proof}
Item (22) is immediate from axioms \axT\ and \axSq. For item~(23) observe that $\square\sneg\square\psi \in\Delta$ implies that $\sneg \square \psi \in \Delta$, by \axT. Conversely, by contraposition on axiom \axSc\ it follows that  $\sneg \square \psi \to \sneg\sneg \square\sneg \square \psi \in \Delta$ and then, by \cpl, $\sneg \square \psi \to \square\sneg \square \psi \in \Delta$. Hence, if  $\sneg \square \psi \in \Delta$ then $\square\sneg \square \psi \in \Delta$. For item~(24) observe that, by contraposition and axiom \axB, 
$\sneg\psi \to \sneg\sneg\square\sneg \square \psi \in \Delta$. By \cpl, $\sneg\psi \to \square\sneg \square \psi \in \Delta$. From this, if   $\psi \notin\Delta$ then $\square\sneg \square \psi \in \Delta$. Items~(25) and~(26) are proved by suitably combining the previous ones.
\end{proof}

\begin{proposition} [Truth lemma] \label{TruthL-GEN} Let ${\bf L} \in \mathbb{L}$, and let $\Delta$ be a $\delta$-saturated set in {\bf L}. Let $v_\Delta:For(\Sigma) \to B$ be a function defined as follows: $v_\Delta(\psi)=a$ iff $\theta_a(\psi) \in \Delta$. Then, $v_\Delta$ is a well-defined function, and it is a valuation over the Nmatrix $\mathcal{M}_{\bf L}$ such that, for every formula $\psi$, $v_\Delta(\psi) \in \textrm{D}$ iff $\psi \in \Delta$.
\end{proposition}
\begin{proof} By Proposition~\ref{TruthL}, it is enough to prove that $v_\Delta(\square \varphi) \in \square_{\bf L} v_\Delta(\varphi)$, for each  {\bf L}. Thus, let  $v_\Delta(\varphi)=a$ and $v_\Delta(\square \varphi)=b$.
\begin{itemize}
\item[-] ${\bf L}=\INTq$. As in the proof of Proposition~\ref{TruthL} it follows that $\gamma_1(\square \varphi)=\square \varphi=\gamma_2(\varphi)$, hence $b_1= a_2$. Since $\gamma_2(\square\varphi)=\square\square\varphi$ and $\gamma_2(\varphi)= \square\varphi$ then, by Proposition~\ref{Delta-satT-GEN}(22), $b_2=a_2$.
\item[-] ${\bf L}=\INTs$. By the previous item, $b_1=a_2$.  Since $\gamma_3(\square \varphi)= \square \sneg \square\varphi$ and $\gamma_2(\varphi)= \square\varphi$ then $b_3= \sneg a_2$, by Proposition~\ref{Delta-satT-GEN}(23).
\item[-] ${\bf L}=\INTB$. As above, $b_1=a_2$. Since $\gamma_3(\square \varphi)= \square \sneg\square\varphi$ and $\gamma_1(\varphi)=\varphi$ then $\sneg a_1 \leq b_3$, by Proposition~\ref{Delta-satT-GEN}(24).
\item[-] ${\bf L}=\INTqB$. By the cases \INTq\ and \INTB\ it follows that  $b_1=b_2=a_2$ and  $\sneg a_1 \leq b_3$.
\item[-] ${\bf L}=\INTqs$. By the cases \INTq\ and \INTs\ it follows that $b_1=b_2=a_2$ and  $b_3= \sneg a_2$.
\end{itemize}
Then, for every {\bf L} as above it follows that  $v_\Delta(\square \varphi) \in \square_{\bf L} v_\Delta(\varphi)$, for every $\varphi$. Hence, the function $v_\Delta$ is a valuation over the Nmatrix $\mathcal{M}_{\bf L}$ such that, by definition, $v_\Delta(\psi) \in \textrm{D}$ iff $\psi \in \Delta$, for every formula $\psi$.
\end{proof}

\begin{theorem} [Completeness]
Let ${\bf L} \in \mathbb{L}$ and let $\Gamma \cup \{\varphi\}$ be a set of formulas. Then, $\Gamma \models_{\bf L} \varphi$ implies that $\Gamma \vdash_{\bf L} \varphi$.
\end{theorem}
\begin{proof}
Suppose that $\Gamma \nvdash_{\bf L} \varphi$. Then, by Remarks~\ref{saturated-rem}(2), there exists a $\varphi$-saturated set $\Delta$ in {\bf L} such that $\Gamma \subseteq \Delta$. Let $v_\Delta$ be the valuation over the Nmatrix $\mathcal{M}_{\bf L}$ as  in Proposition~\ref{TruthL-GEN}. Then, $v_\Delta(\psi) \in \textrm{D}$ for every $\psi \in \Gamma$, but $v_\Delta(\varphi) \notin \textrm{D}$. That is, $\Gamma \not\models_{\bf L} \varphi$.
\end{proof}

\noindent
Finally, by an easy adaptation of the corresponding proofs given in Section~\ref{conserva}, it can be proven that  the logics obtained by mixing the extensions of \Tm\ with \pyn\ can be seen, in fact, as obtained by a process of combination of logics: indeed, each of such combined systems conservatively expands its components.

\begin{theorem} [Conservativity]
Let ${\bf L} \in \mathbb{L}$. Then, {\bf L} conservatively expands \pyn\ and the Ivlev-like modal logics on which it is based. That is:\\
(1) For every set of formulas $\Gamma \cup \{\varphi\} \subseteq For(\Sigma_4)$: \  $\Gamma \vdash_{\bf L} \varphi$ \ iff \  $\Gamma \vdash_{\pyn} \varphi$.\\
(2) For every set of formulas $\Gamma \cup \{\varphi\}  \subseteq For(\Sigma_m)$:\\
\indent (2.1) $\Gamma \vdash_{\INTq} \varphi$ \ iff \  $\Gamma \vdash_{\Tqm} \varphi$.\\
\indent  (2.2) $\Gamma \vdash_{\INTqs} \varphi$ \ iff \  $\Gamma \vdash_{\Tqsm} \varphi$.\\
\indent  (2.3) $\Gamma \vdash_{\INTB} \varphi$ \ iff \  $\Gamma \vdash_{\TBm} \varphi$.\\
\indent  (2.4) $\Gamma \vdash_{\INTqB} \varphi$ \ iff \  $\Gamma \vdash_{\TqBm} \varphi$.\\
\indent  (2.5) $\Gamma \vdash_{\INTs} \varphi$ \ iff \  $\Gamma \vdash_{\Tsm} \varphi$.
\end{theorem}

\section{Final remarks} \label{finalsect}

In this paper we introduced a methodology for combining logics presented by means of (finite) swap structure semantics, by mixing its snapshots (truth-values). The case example analized here was the combination of \pyn, an interesting expansion of \fde\ by an implication, and several 4-valued Ivlev-like modal logics. While the former is presented by means of a 4-valued logical matrix expressed by a twist structure (a particular case of swap structure), each one of the later is presented by a 4-valued Nmatrix expressed by a swap structure. Then, a 4-element twist structure was combined with a 4-element swap structure, producing a 6-element swap structure (Nmatrix). As commented previously, the method would produce 8 snapshots by superposition, but some restrictions were imposed by conceptual reasons, thus producing 6 snapshots.

A natural question we might ask is to what extent the Ivlev-type modal logics considered here  could {\em in fact}  be considered as modal logics. In~\cite[Theorems~24 and~29]{con:cerro:per:15} (together with the corrected proof of Theorem~2 given in~\cite{con:cerro:per:16}) it was proved that the extension of \Tm\ by adding the necessitation rule \nec\ coincides exactly with the standard normal modal system {\bf KT}. Moreover, it was also proven (taking into consideration the small technical corrections given in~\cite{con:cerro:per:16}) that the same holds for the other Ivlev-like modal logics considered in this paper, that is: if {\bf Lm} is the (non-normal) Ivlev-like version of normal modal logic {\bf L}, then the extension of {\bf Lm} by adding \nec\ coincides exactly with {\bf L}.\footnote{To be strict, in~\cite{con:cerro:per:15} this result was proved neither for \TqBm\  nor for \Tsm, given that these logics were originally proposed in the present paper. However, the proof that  the extension of \TqBm\ and  \Tsm\ by adding \nec\ coincides respectively with the normal modal system {\bf KT4B} and  {\bf KT5} can be easily obtained from the proof for the other systems.} This is a good reason to consider Ivlev-like logics as genuine non-normal modal logics (with an hyperintensional modality). They can be seen  as specific \nec-less versions of normal modal systems.

It is important to observe that \pyn\ can be characterized by a class of twist structures, one for each classical implicative lattice. In turn, four of the extensions of \Tm\ considered here can be characterized by a class of swap structures, one for each Boolean algebra (see~\cite{con:gol:19}), and the two new extensions \TqBm\ and \Tsm\  proposed here can be treated analogously. By the very nature of the process of combining swap structures proposed here, in terms of formal specifications, it is clear that the corresponding 6-valued characteristic Nmatrix for each of the 6 combined paradefinite modal logics can be extended to  a class of swap structures defined over Boolean algebras. The reason for obtaining Boolean algebras in the general case is that the components of the snapshots must be defined over classical implicative lattices with bottom, which are Boolean algebras.

The notion of swap structures for \Tm\ and its extensions can be defined in a more precise way by means of bivaluations, as it was done in~\cite[Section~6.4]{car:con:16}, \cite[Section~5.2]{con:tol:22} and~\cite{con:rod:22}. Consider, for instance, the case of \Tm. The set of snapshots is the set of triples over $\{0,1\}$ of the form  $(b(\varphi),b(\square\varphi),b(\square\sneg\varphi))$ for any formula $\varphi$ and any bivaluation $b$ for \Tm, which produces exactly the domain $A_m$. A bivaluation for \Tm\ is a function $b:For(\Sigma_m) \to \{0,1\}$  satisfying the properties of the classical 2-valued  valuations for \cpl\ over $\{\to,\sneg\}$ plus the restrictions obtained from the formal specifications of the swap structures for \Tm\ (or from the original axioms of \Tm), namely: $b(\square\varphi) \leq b(\varphi)$; 
$b(\square\sneg\varphi) \sqcup b(\square\psi) \leq  b(\square(\varphi \to\psi)) \leq (b(\square\varphi) \Rightarrow b(\square\psi)) \sqcap (b(\square\sneg\psi) \Rightarrow b(\square\sneg\varphi))$; $b(\square\sneg(\varphi \to\psi))=b(\square\varphi) \sqcap b(\square\sneg\psi)$; and $b(\square\sneg\sneg\varphi)=b(\square\varphi)$.

The class of such bivaluations is a sound and complete semantics for \Tm, by considering $1$ as the designated value. The relationship between bivaluations for \Tm\ and valuations over $\mathcal{M}_T$ is the following: given $b$, the function $v(\varphi)=(b(\varphi),b(\square\varphi),b(\square\sneg\varphi))$ is a valuation over $\mathcal{M}_T$ such that $b(\varphi)=1$ iff $v(\varphi) \in D_m$. Conversely, given $v$, the function $b(\varphi)=v(\varphi)_1$ is a bivaluation for \Tm\  such that $b(\varphi)=1$ iff $v(\varphi) \in D_m$. Bivaluations for \pyn\ are mappings  $b:For(\Sigma_4) \to \{0,1\}$  satisfying the properties of the classical 2-valued  valuations for \cpl\ over $\{\land,\vee,\to\}$ plus the following restrictions: $b(\neg(\varphi \land \psi))=b(\neg\varphi) \sqcup b(\neg\psi)$; $b(\neg(\varphi \vee \psi))=b(\neg\varphi) \sqcap b(\neg\psi)$;  $b(\neg(\varphi \to \psi))=b(\varphi) \sqcap b(\neg\psi)$; and $b(\neg\neg\varphi)=b(\varphi)$. Similar to \Tm, there is a one-to-one correspondence between valuations $v$ over the matrix $\mathcal{M}_4$ and bivaluations $b$ for \pyn: given $v$, define $b(\varphi)=v(\varphi)_1$. Conversely, given $b$, define $v(\varphi)=(b(\varphi),b(\neg\varphi))$. In both cases, $v$ and $b$ satisfy the same formulas. As shown in~\cite{con:tol:22} and~\cite{con:rod:22}, bivaluations can be easily extended to the class of Boolean algebras (or classical implicative lattices, in the case of \pyn). Under this formal perspective based on bivaluations, the combination process proposed here is analogous to consider bivaluations over the mixed signature satisfying the requirements of the bivaluations of the logics being combined (from which it is obtained the swap structures semantics for the combined logic).

The proposed method for combining logics presented by swap structures can be particularized to combinations of swap structures with twist structures (as in the case in the present paper) or even to combinations of twist structures with twist structures, thus opening a wide variety of potential applications. By considering the close relationship between swap (and twist) structures semantics and bivaluations, it would be reasonable to define the combination (by fibring) of bivaluation semantics, giving rise (under certain circumstances) to the combination by fibring of swap structures, which are special cases of Nmatrices.   This is the subject of future research.

Concerning the paradefinite modal systems  obtained here by combination of logics we can mention, as related works, the (normal) modal systems studied by S. Odinstov and H. Wansing in~\cite{odin:wan:10} and~\cite{odin:wan:17}, by combining the minimal normal modal logic {\bf K} with \bfour. As in our case, and in order to harmonize both logics, they consider some bridge principles between $\square$, $\sneg$ and $\neg$, some of them stronger than ours. They propose twist structures semantics defined over  Boolean algebras with modal operators, a.k.a. {\em modal algebras}, given the validity of \nec. Another related approach was proposed previously in~\cite{riv:jun:jan:15}, also using twist structures semantics, but now based on the bilattice structure of \fde\ and using, besides $\to$, the {\em strong implication} $\Longrightarrow$ introduced in the realm of relevant logic by R. Brady (see the matrix logic $M_4$ in~\cite[p. 10]{brady:82}, as well as its Hilbert calculus BN4 on pp. 21--23). This implication was afterwards rediscovered by O. Arieli and A. Avron in~\cite{ari:av:96}, and characterized  as follows: $\varphi \Longrightarrow \psi \defi (\varphi \to \psi) \land (\neg\psi \to \neg\varphi)$. The 4-valued paradefinite modal logics investigated  in~\cite{riv:jun:jan:15} are non-normal. Moreover, besides \nec, axiom \axK\ is also not valid. Before~\cite{riv:jun:jan:15}, and from the perspective of relevant logics, L. Goble proposed in~\cite{gob:06} a normal modal extension of Brady's system BN4. Goble's system was revisited in~\cite{odin:wan:17}. One advantage of using strong implication is that it enjoys contraposition. However, its axiomatization is more complicated, requiring more inference rules besides \MP.  The possibility of extending our approach to BN4 (based on  $\Longrightarrow$) or to other implication connectives over \fde\  will be investigated.


We consider that the paradefinite Ivlev-like systems obtained here are instances of an interesting and promising alternative to the standard approach to modal logics. Indeed, they offer a new perspective on the study of non-normal model systems, on the one hand, and they are decidable by finite-valued Nmatrices, on the other. The potential applications to informational systems, because of its decidability and its paraconsistent and paracomplete character based on \fde\ and \lets, as well as its  expressivity and relative simplicity,  suggest possible applications of this kind of systems to the analysis of modal paradoxes, among other domains.

\

\noindent
{\bf Acknowledgments:}
The main ideas contained in this article were originally presented in~\cite{con:EBL22}. The author acknowledges support from the National Council for Scientific and Technological Development (CNPq,~Brazil), through the individual  research grant \# 309830/2023-0, and the ``Edital Universal''  project \# 408040/2021-1. This research was also supported by the S\~ao Paulo Research Foundation (FAPESP, Brazil) trough the Thematic Project {\em Rationality, logic and  probability -- RatioLog}, grant \#2020/16353-3.


\


\

\section{Appendix: A Prolog program for \INT\ and its extensions}

\scriptsize
\begin{verbatim} 
%%%%%%%%%%%%%%% Boolean operators %%%%%%%%%%%%%%%

andb(0,1,0).    %%% (0 AND 1) = 0 etc.
andb(0,0,0).
andb(1,0,0).
andb(1,1,1).
orb(0,1,1).     %%% (0 OR 1) = 1 etc.
orb(0,0,0).
orb(1,0,1).
orb(1,1,1).
impb(0,1,1).    %%% (0 IMPLIES 1) = 1 etc.
impb(0,0,1).
impb(1,0,0).
impb(1,1,1).
notb(1,0).      %%% NOT(1) = 0 etc.
notb(0,1).


%%%%%%%%%%%%%%% Snapshots definition %%%%%%%%%%%%%%%

boo(1).
boo(0).

sn(Z1,Z2,Z3,Z4) :- boo(Z1),boo(Z2),boo(Z3),boo(Z4),andb(Z1,Z2,Z2),
    andb(Z1,Z3,0), andb(Z2,Z4,0), andb(Z3,Z4,Z3).


%%%%%%%%%%%%%%% Multioperators for the Nmatrices %%%%%%%%%%%%%%%

conj(Z1,Z2,Z3,Z4,W1,W2,W3,W4,U1,U2,U3,U4) :- sn(Z1,Z2,Z3,Z4),sn(W1,W2,W3,W4),
   andb(Z1,W1,U1),andb(Z2,W2,U2),orb(Z3,W3,M1),
   impb(Z2,W3,M3),impb(W2,Z3,M4),andb(M3,M4,N2),
   andb(M1,U3,M1),andb(U3,N2,U3),orb(Z4,W4,U4),sn(U1,U2,U3,U4).

disj(Z1,Z2,Z3,Z4,W1,W2,W3,W4,U1,U2,U3,U4) :- sn(Z1,Z2,Z3,Z4),sn(W1,W2,W3,W4),
   orb(Z1,W1,U1),andb(Z3,W3,U3),andb(Z4,W4,U4),orb(Z2,W2,M1),
   impb(Z3,W2,M3),impb(W3,Z2,M4),andb(M3,M4,N2),andb(M1,U2,M1),
    andb(U2,N2,U2),sn(U1,U2,U3,U4).

imp(Z1,Z2,Z3,Z4,W1,W2,W3,W4,U1,U2,U3,U4) :- sn(Z1,Z2,Z3,Z4),sn(W1,W2,W3,W4),
   impb(Z1,W1,U1),andb(Z2,W3,U3),andb(Z1,W4,U4),orb(Z3,W2,M1),
   impb(Z2,W2,M3),impb(W3,Z3,M4),andb(M3,M4,N2), 
   andb(M1,U2,M1),andb(U2,N2,U2),sn(U1,U2,U3,U4).

neg(Z1,Z2,Z3,Z4,Z4,Z3,Z2,Z1).

sneg(Z1,Z2,Z3,Z4,U1,Z3,Z2,U4) :- sn(Z1,Z2,Z3,Z4), notb(Z1,U1),
   andb(Z2,U4,Z2),andb(U4,Z3,0).

box(Z1,Z2,Z3,Z4,Z2,U2,U3,U4) :- sn(Z1,Z2,Z3,Z4), sn(Z2,U2,U3,U4).
box4(Z1,Z2,Z3,Z4,Z2,Z2,U3,U4) :- sn(Z1,Z2,Z3,Z4), sn(Z2,Z2,U3,U4).
box45(Z1,Z2,Z3,Z4,Z2,Z2,U3,U4) :- sn(Z1,Z2,Z3,Z4), orb(U3,Z2,1),sn(Z2,Z2,U3,U4).
box5(Z1,Z2,Z3,Z4,Z2,U2,U3,U4) :- sn(Z1,Z2,Z3,Z4),orb(U3,Z2,1),sn(Z2,U2,U3,U4).
boxb(Z1,Z2,Z3,Z4,Z2,U2,U3,U4) :- sn(Z1,Z2,Z3,Z4), orb(U3,Z1,1),sn(Z2,U2,U3,U4).
box4b(Z1,Z2,Z3,Z4,Z2,Z2,U3,U4) :- sn(Z1,Z2,Z3,Z4), orb(U3,Z1,1),sn(Z2,Z2,U3,U4).


%%%%%%%%%%%%%%% List operators %%%%%%%%%%%%%%%

append([],L,L).
append([A|L1],L2,[A|L]) :- append(L1,L2,L).

member(X,[X|_]) :- !.
member(X,[_|L]) :- member(X,L).


%%%%%%%%%%%%%%% Truth-tables generator %%%%%%%%%%%%%%%

% Usage: valu(L,[],LF), where L is the list of formulas to be placed in the truth-table, 
% starting from the propositional variables (in this case only p1 and p2 were declared, 
% but it can be easily extended). The formulas in L must be placed in order of appearence
% of subformulas (as it is done when construct a truth-table) and LF returns the list
% of lists [A,U1,U2,U3,U4], where A is a formula of L and (U1,U2,U3,U4) is its truth-value in LF
% Syntax for writting formulas: p1,p2 (propositional variables)
% and(A,B), or(A,B), implies(A,B), snon(A) (classical negation), non(A) (FDE negation)
% tbox(A), tbox4(A), tbox5(A), tboxb(A) (system B), tbox45(A), tbox4b(A) (boxes)

valu([],L,L).
valu([p1|L0],L1,L) :- sn(U1,U2,U3,U4),append(L1,[[p1,U1,U2,U3,U4]],L2),valu(L0,L2,L).
valu([p2|L0],L1,L) :- sn(U1,U2,U3,U4),append(L1,[[p2,U1,U2,U3,U4]],L2), valu(L0,L2,L).
% if necessary, insert analogous clauses for p3, p4, etc

valu([and(A,B)|L0],L1,L) :-
   member([A,Z1,Z2,Z3,Z4],L1),member([B,W1,W2,W3,W4],L1),
   conj(Z1,Z2,Z3,Z4,W1,W2,W3,W4,U1,U2,U3,U4),
   append(L1,[[and(A,B),U1,U2,U3,U4]],L2),valu(L0,L2,L).
valu([or(A,B)|L0],L1,L) :-
   member([A,Z1,Z2,Z3,Z4],L1),member([B,W1,W2,W3,W4],L1),
   disj(Z1,Z2,Z3,Z4,W1,W2,W3,W4,U1,U2,U3,U4),
   append(L1,[[or(A,B),U1,U2,U3,U4]],L2),valu(L0,L2,L).
valu([implies(A,B)|L0],L1,L) :-
   member([A,Z1,Z2,Z3,Z4],L1),member([B,W1,W2,W3,W4],L1),
   imp(Z1,Z2,Z3,Z4,W1,W2,W3,W4,U1,U2,U3,U4),
   append(L1,[[implies(A,B),U1,U2,U3,U4]],L2),valu(L0,L2,L).
valu([tbox(A)|L0],L1,L) :-
   member([A,Z1,Z2,Z3,Z4],L1),box(Z1,Z2,Z3,Z4,U1,U2,U3,U4),
   append(L1,[[tbox(A),U1,U2,U3,U4]],L2),valu(L0,L2,L).
valu([tbox4(A)|L0],L1,L) :-
   member([A,Z1,Z2,Z3,Z4],L1),box4(Z1,Z2,Z3,Z4,U1,U2,U3,U4),
   append(L1,[[tbox4(A),U1,U2,U3,U4]],L2),valu(L0,L2,L).
valu([tbox5(A)|L0],L1,L) :-
   member([A,Z1,Z2,Z3,Z4],L1),box5(Z1,Z2,Z3,Z4,U1,U2,U3,U4),
   append(L1,[[tbox5(A),U1,U2,U3,U4]],L2),valu(L0,L2,L).
valu([tbox45(A)|L0],L1,L) :-
   member([A,Z1,Z2,Z3,Z4],L1),box45(Z1,Z2,Z3,Z4,U1,U2,U3,U4),
   append(L1,[[tbox45(A),U1,U2,U3,U4]],L2),valu(L0,L2,L).
valu([tboxb(A)|L0],L1,L) :-
   member([A,Z1,Z2,Z3,Z4],L1),boxb(Z1,Z2,Z3,Z4,U1,U2,U3,U4),
   append(L1,[[tboxb(A),U1,U2,U3,U4]],L2),valu(L0,L2,L).
valu([tbox4b(A)|L0],L1,L) :-
   member([A,Z1,Z2,Z3,Z4],L1),box4b(Z1,Z2,Z3,Z4,U1,U2,U3,U4),
   append(L1,[[tbox4b(A),U1,U2,U3,U4]],L2),valu(L0,L2,L).
valu([snon(A)|L0],L1,L) :-
   member([A,Z1,Z2,Z3,Z4],L1),sneg(Z1,Z2,Z3,Z4,U1,U2,U3,U4),
   append(L1,[[snon(A),U1,U2,U3,U4]],L2),valu(L0,L2,L).
valu([non(A)|L0],L1,L) :-
   member([A,Z1,Z2,Z3,Z4],L1),neg(Z1,Z2,Z3,Z4,U1,U2,U3,U4),
   append(L1,[[non(A),U1,U2,U3,U4]],L2),valu(L0,L2,L).


%%%%%%%%%%%%%%% Checking validity, satisfiability and refutability %%%%%%%%%%%%%%%

% check if the last formula A of L is refutable, showing in L1 a countermodel for A
% (it shows all of them)
refutable(L,L1) :- valu(L,[],L1),last(L1,[_,0,_,_,_]). 

% check if the last formula A of L is satisfiable, showing in L1 a model for A
% (it shows all of them)
satisfiable(L,L1) :- valu(L,[],L1),last(L1,[_,1,_,_,_]).

% unsatisfiability test: "true" iff the last formula of L is unsatisfiable
unsatisfiable(L) :- not(satisfiable(L,_)).

% validity test: "true" iff the last formula of L is valid
valid(L) :- not(refutable(L,_)).


%%%%%%%%%%%%%%% Example: Axiom B %%%%%%%%%%%%%%%

% The query below constructs a truth-table for axiom B in IvFDE_B: 
% valu([p1,tboxb(p1),snon(tboxb(p1)),tboxb(snon(tboxb(p1))),
% snon(tboxb(snon(tboxb(p1)))),implies(snon(tboxb(snon(tboxb(p1)))),p1)],[],L).
% The query valid(L1), where L1 is the leftmost list in the query above,
% returns "true" (but returns "false" if tboxb is changed to tbox or tbox4).
% It also returns "true" by changing tboxb by tbox5, proving that axiom B also holds in IvFDE_5.
% By using tbox then satisfiable(L1,L2) and refutable(L1,L3) will return
% all the models and countermodels of B in IvFDE_T; idem in IvFDE_4 by using tbox4
\end{verbatim}

\end{document}